\documentclass[11pt]{article}

\def\Anonymity{0}%

\usepackage{geometry}
\usepackage[T1]{fontenc}
\usepackage{graphicx}
\usepackage{palatino}
\usepackage{inconsolata}
\usepackage{xspace}
\usepackage[utf8]{inputenc}
\usepackage{amsmath, amssymb, amsfonts, amsthm,mathtools}
\usepackage{enumerate}
\usepackage{xcolor}
\usepackage[breaklinks,colorlinks,linkcolor=magenta,citecolor=blue,bookmarks,bookmarksopen,bookmarksnumbered]{hyperref}
\usepackage{bm}
\usepackage{paralist}
\usepackage{thmtools}
\usepackage{rotating}
\usepackage{mdframed}
\usepackage{multirow}
\usepackage{fancyvrb}
\usepackage{mathrsfs}
\usepackage{indentfirst}
\usepackage{algorithm, algorithmicx, algpseudocodex}
\usepackage{mleftright}
\usepackage[nottoc]{tocbibind}%
\usepackage[langfont=roman]{complexity}
\usepackage{tablefootnote}
\usepackage{relsize}
\usepackage{exscale}
\usepackage{tikz}
\usepackage{afterpage}

\usepackage[backend=biber, style=alphabetic, backref=true, maxbibnames=99, maxalphanames=4, minalphanames=3]{biblatex}

\newbibmacro{string+doiurl}[1]{%
  \iffieldundef{doi}{%
    \iffieldundef{url}{%
      #1%
    }{%
      \href{\thefield{url}}{#1}%
    }%
  }{%
    \href{https://doi.org/\thefield{doi}}{#1}%
  }%
}

\DeclareFieldFormat{title}{\usebibmacro{string+doiurl}{\mkbibemph{#1}}}
\DeclareFieldFormat[article,incollection,inproceedings,thesis]{title}{\usebibmacro{string+doiurl}{\mkbibquote{#1}}}

\DeclareFieldFormat
  [article,inbook,incollection,inproceedings,patent,thesis,unpublished]
  {titlecase}{%
    \ifcurrentfield{title}{\MakeSentenceCase*{#1}}{%
      \ifcurrentfield{subtitle}{\MakeSentenceCase*{#1}}{#1}%
    }%
  }

\usepackage[capitalize]{cleveref}%

\allowdisplaybreaks

\newtheorem{theorem}{Theorem}[section]
\newtheorem{lemma}[theorem]{Lemma}
\newtheorem{corollary}[theorem]{Corollary}
\newtheorem{claim}[theorem]{Claim}

\newtheorem{question}[theorem]{Question}

\theoremstyle{definition}
\newtheorem{definition}[theorem]{Definition}

\renewcommand{\MAEXP}{\mathsf{MA}_{\mathsf{EXP}}}
\renewcommand{\AMEXP}{\mathsf{AM}_{\mathsf{EXP}}}

\theoremstyle{remark}
\newtheorem{remark}[theorem]{Remark}

\newenvironment{claimproof}[1][\proofname]
{\proof[#1]}
{\endproof}

\algtext*{EndWhile}%
\algtext*{EndIf}%
\algtext*{EndFor}%
\algtext*{EndFunction}%

\algnewcommand{\IfThen}[2]%
{\State \algorithmicif\ #1\ \algorithmicthen\ #2}

\makeatletter
\def\moverlay{\mathpalette\mov@rlay}
\def\mov@rlay#1#2{\leavevmode\vtop{%
		\baselineskip\z@skip \lineskiplimit-\maxdimen
		\ialign{\hfil$\m@th#1##$\hfil\cr#2\crcr}}}
\newcommand{\charfusion}[3][\mathord]{
	#1{\ifx#1\mathop\vphantom{#2}\fi
		\mathpalette\mov@rlay{#2\cr#3}
	}
	\ifx#1\mathop\expandafter\displaylimits\fi}
\makeatother

\renewcommand{\poly}{\mathrm{poly}}

\renewcommand{\polylog}{\mathrm{polylog}}

\newcommand{\eps}{\varepsilon}

\def\BF{\mathsf{BF}}

\newlang{\MCSP}{MCSP}
\newlang{\MFSP}{MFSP}
\newlang{\MKtP}{MKtP}
\newlang{\MKTP}{MKTP}
\newlang{\itrMCSP}{itrMCSP}
\newlang{\itrMKTP}{itrMKTP}
\newlang{\itrMINKT}{itrMINKT}
\newlang{\MINKT}{MINKT}
\newlang{\MINK}{MINK}
\newlang{\MINcKT}{MINcKT}
\newlang{\CMD}{CMD}
\newlang{\DCMD}{DCMD}
\newlang{\CGL}{CGL}
\newlang{\PARITY}{PARITY}
\newlang{\SGV}{SGV}
\renewlang{\Gap}{Gap}
\newlang{\Avoid}{\textnormal{\textsc{Avoid}}}
\newlang{\MissingString}{\textsc{Missing-String}}
\newlang{\SinkOfDAG}{\textsc{Sink-Of-DAG}}
\newlang{\Iter}{\textsc{Iter}}
\newlang{\Palindromes}{\textsc{Palindromes}}
\newlang{\Sparsification}{\textsc{Sparsification}}
\newlang{\HamEst}{\mathsf{HammingEst}}
\newlang{\HamHit}{\mathsf{HammingHit}}
\newlang{\CktEval}{\textsc{Circuit-Eval}}
\newlang{\Hard}{\textsc{Hard}}
\newlang{\cHard}{\textsc{cHard}}
\newlang{\CAPP}{CAPP}
\newlang{\GapUNSAT}{GapUNSAT}
\newlang{\searchSAT}{searchSAT}
\newlang{\OV}{OV}
\newlang{\PRIMES}{PRIMES}
\renewlang{\PCP}{PCP}
\newlang{\PCPP}{PCPP}
\newlang{\CircuitSAT}{\textnormal{\textsc{Circuit-SAT}}}
\newclass{\FMA}{FMA}
\newclass{\Avg}{Avg}
\newclass{\ZPEXP}{ZPEXP}
\newclass{\DLOGTIME}{DLOGTIME}
\newclass{\ALOGTIME}{ALOGTIME}
\newclass{\MATIME}{MATIME}
\newclass{\ATIME}{ATIME}%
\newclass{\SZKA}{SZKA}
\newclass{\Laconic}{Laconic\text{-}}
\newclass{\APEPP}{APEPP}
\newclass{\SAPEPP}{SAPEPP}
\newclass{\TFSigma}{TF\Sigma}
\newclass{\NTIMEGUESS}{NTIMEGUESS}
\newclass{\FZPP}{FZPP}
\newclass{\UEoPL}{UEoPL}
\newclass{\EoPL}{EoPL}
\newclass{\SoPL}{SoPL}
\newclass{\CLS}{CLS}
\newclass{\PWPP}{PWPP}
\newclass{\SVN}{SVN}

\newlang{\Formula}{Formula}
\newlang{\THR}{THR}
\newlang{\MAJ}{MAJ}
\newlang{\DOR}{DOR}
\newlang{\ETHR}{ETHR}
\newlang{\Midbit}{Midbit}
\newlang{\LCS}{LCS}
\newlang{\TAUT}{TAUT}

\newcommand{\pr}{\textnormal{\textrm{pr}}}
\newcommand{\prMA}{{\pr\MA}}%
\newcommand{\prAM}{{\pr\AM}}

\newcommand{\prMATIME}{{\pr\MATIME}}

\newcommand{\smart}{\mathsf{smart}\text{-}}

\renewcommand{\S}{\mathsf{S}}%

\newcommand{\calA}{\mathcal{A}}
\newcommand{\calB}{\mathcal{B}}
\newcommand{\calC}{\mathcal{C}}
\newcommand{\calD}{\mathcal{D}}

\newcommand{\N}{\mathbb{N}}

\newcommand{\F}{\mathbb{F}}

\newcommand{\pfx}{\mathsf{pfx}}

\newcommand{\Yes}{\textsc{Yes}}
\newcommand{\No}{\textsc{No}}

\newcommand{\Corr}{\mathsf{Corr}}%

\newcommand{\Range}{\mathrm{Range}}

\newcommand{\JK}{\textnormal{\textsf{Je\v{r}\'{a}bek--Korten}}}

\newcommand{\Jerabek}{Je\v{r}\'{a}bek\xspace}

\usepackage[inline,marginclue,index]{fixme}  %
\fxsetup{theme=color,mode=multiuser}

\definecolor{color1}{RGB}{46,134,193}
\FXRegisterAuthor{hanlin}{ahanlin}{\colorbox{color1}{\color{white}Hanlin}}

\FXRegisterAuthor{ryan}{aryan}{\colorbox{cyan}{\color{white}Ryan}}

\usepackage{aliascnt} %

\makeatletter
\newcommand{\ALG@lineautorefname}{Line}
\makeatother

\bibliography{ref}
\def\Solve{\mathsf{Solve}}
\def\Recon{\mathsf{Recon}}

\newcommand{\RoundLength}[2]{\ensuremath{\mleft[\scalebox{0.9}{$
    \begin{aligned}
        \textnormal{\textsf{\#rounds}}=&~{#1}\\
        \textnormal{\textsf{length}}=&~{#2}
    \end{aligned}
$}\mright]}}

\begin{document}

\newgeometry{margin=0.8in}
\title{Near-Maximum Circuit Lower Bounds for Exponential Time with Merlin-Arthur Queries}
\ifnum\Anonymity=0
\author{
	Hanlin Ren\footnote{Hanlin Ren is supported by the Massive Dynamics Member Fund at the Institute for Advanced Study.}\\ \small{IAS} \\ \small{\texttt{\href{mailto:h4n1in.r3n@gmail.com}{h4n1in.r3n@gmail.com}}}
	\and
    Ryan Williams\footnote{This work was initiated while visiting the Institute for Advanced Study, Princeton, NJ. This material is based upon work supported by the National Science Foundation under grants DMS-2424441 (at IAS) and CCF-2420092 (at MIT).}\\ \small{MIT}\\
    \small{\texttt{\href{mailto:rrw@mit.edu}{rrw@mit.edu}}}
}
\fi

\maketitle
\vspace{-2em}
\pagenumbering{gobble}

\begin{abstract}
    We prove a near-maximum ($2^n / n$) circuit lower bound for the complexity class $\E^\prMA/_1$, corresponding to exponential time with access to a promise-$\MA$ oracle and one bit of advice. Our proof incorporates the iterative win-win paradigm (Chen--Lu--Oliveira--Ren--Santhanam, FOCS'23), the reduction from the Range Avoidance problem to circuit lower bounds (\Jerabek, \emph{Ann.~Pure Appl.~Log.~'04}; Korten, FOCS'21), and the PCP theorem. Crucial to our proof is the analysis of the complexity class $\P^\NP\RoundLength{r}{s}$, which is $\P^\NP$ with $r(n)$ adaptive rounds of $\NP$ queries, where each $\NP$ query has witness length $s(n)$.
\end{abstract}

{\tableofcontents}

\newpage

\pagenumbering{arabic}

\section{Introduction}

A simple counting argument dating back to Shannon~\cite{Shannon49} shows that a uniformly random $2^n$-bit truth table describes a \emph{near-maximum hard} function requiring Boolean circuits of size $\Omega(2^n/n)$ on inputs of length $n$, with high probability: there are simply too many possible functions to be covered by all circuits of size $2^n/(10n)$. Significant research has gone into understanding how efficiently such a hard function can be constructed, and this problem is central to complexity theory. For example, if there are functions in $\NP$ with near-maximum hard (or even superpolynomially hard) finite slices, then $\P \neq \NP$. If there are near-maximum hard (or even exponentially-hard) functions in $\E = \TIME[2^{O(n)}]$, then $\P = \BPP$~\cite{IW97}. However, the \emph{non-uniformity} of circuit families makes the uniform construction of hard functions apparently extremely difficult. It has been known for decades that the (huge) class $\Sigma_3 \E = \Sigma_3 \TIME[2^{O(n)}]$ contains a function of near-maximum circuit size~\cite{Kannan82} and that the complexity can be reduced slightly to $\E^{\Sigma_2 \P}$~\cite{MVW99}. 

In the last few years, substantial progress has been made on understanding the uniform complexity of near-maximum hard functions. Chen, Hirahara, and Ren~\cite{CHR24} showed that $\Sigma_2 \E$ (indeed, even $\S_2 \E/_1$) contains near-maximum hard functions, resolving a 40-year open problem, using a novel win-win paradigm for constructing hard functions. Li~\cite{Li24} gave a drastically simplified proof that also improved the lower bound in several aspects; see \cite{CHLR26} for a combined exposition. More recently, Chen, Li, and Liang~\cite{CLL25} showed that exponential time Arthur-Merlin with subexponential advice ($\AMEXP/_{2^{n^{\eps}}}$) contains near-maximum hard functions. 

These recent results raise the question: can the complexity of exponentially hard functions be reduced even further? Natural candidates for such functions are those in complexity classes where we already know some nontrivial circuit complexity lower bounds, but the known lower bounds are much smaller than $\Omega(2^n/n)$ size. For example, it has been known since \cite{BFT98} that the class \emph{exponential-time Merlin Arthur} ($\MAEXP$) has functions without polynomial-size circuits. Can the complexity of the hard function in \cite{CLL25} be reduced from Arthur-Merlin down to Merlin-Arthur (non-interactive proofs), and can the $2^{n^{\eps}}$ size advice be removed or made much shorter? (Recall that Merlin-Arthur can always be simulated by Arthur-Merlin with polynomial overhead~\cite{Babai85}, but showing the converse is a wide open problem.) An additional motivation for proving circuit lower bounds for classes related to Merlin-Arthur instead of Arthur-Merlin is their connections to \emph{easy witness lemmas}~\cite{ImpagliazzoKW02} and \emph{the algorithmic method}~\cite{MurrayW20, Chen25, Chen23}.

Although Merlin-Arthur complexity classes have long been known to admit non-trivial circuit lower bounds via non-relativizing techniques~\cite{BFT98,Santhanam09}, it has been an outstanding open problem to find a Merlin-Arthur-type complexity class which contains a near-maximum hard function. The proof of \cite{BFT98} that $\MAEXP \not\subset \P/_\poly$ generalizes to show a ``half-exponential'' size lower bound: a lower bound for size functions $h$ such that $h(h(\poly(n))) \leq 2^n$. Their lower bound provably does not relativize: they give an oracle relative to which $\MAEXP$ has small circuits. However, for the slightly larger class $\E^{\pr\MA}$ (exponential time with queries to a promise version of Merlin-Arthur), there \emph{is} a relativizing proof of a half-exponential circuit complexity lower bound (see \autoref{appendix: E to prMA halfexp}) based on \cite{ChakaravarthyR11}. Half-exponential functions are a typical barrier for circuit size lower bounds: prior to the breakthrough of \cite{CHR24}, the best circuit size lower bound for $\Sigma_2 \E$ was only half-exponential~\cite{MVW99}.

\subsection{Our Results} Recall that $\pr\MA$ is the class of promise problems solvable in polynomial time with a Merlin-Arthur protocol. In particular, a promise problem~\cite{EvenSY84} is given by a pair $(\Pi_\Yes, \Pi_\No)$ where $\Pi_\Yes$ and $\Pi_\No$ are disjoint but do not necessarily cover all possible inputs, and the Merlin-Arthur protocol is only judged on its behavior for inputs in $\Pi_\Yes \cup \Pi_\No$ (we say such inputs ``satisfy the promise'').

Our main result is that there is a near-maximum hard function in $\E^{\pr\MA}/_1$, the class of problems computable in exponential time with oracle access to a $\pr\MA$ language and $1$ bit of advice. 

\begin{theorem}[Main Result]\label{thm: main} There is a function in $\E^{\pr\MA}/_1$ that requires  circuits of $\Omega(2^n/n)$ size.
\end{theorem}

In fact, the hard function given in  \autoref{thm: main} has the additional nice property that, on the correct advice bit, the exponential-time algorithm \emph{only} queries the $\pr\MA$ oracle on queries that satisfy the promise of the $\pr\MA$ problem. Such reductions are called \emph{smart reductions}~\cite{GrollmannS88, Goldreich06a}; hence, our lower bound \autoref{thm: main} holds for a function in the potentially smaller class $\smart\E^{\pr\MA}/_1$. We refer the reader to \autoref{sec: smart reductions} for a more formal definition.

Prior work from about 15 years ago hinted that a result like \autoref{thm: main} is possible. Ayd{\i}nl{\i}o\u{g}lu, Gutfreund, Hitchcock, and Kawachi~\cite{AydinliogluGHK11} proved that if $\pr\AM$ is in polynomial time with an $\NP$ oracle ($\prAM \subseteq \P^{\NP}$) then $\E^\NP$ contains near-maximum hard functions. Roughly speaking, they showed how to construct a hard truth table by using a $\prAM$ oracle to approximately count how many circuits agree with the current truth table prefix, allowing them to eliminate circuits efficiently. However, they also pointed out that their result does \emph{not} imply an unconditional circuit lower bound for $\E^{\prAM}$: the hard function they produce \emph{depends on oracle answers outside of the promise}. In that sense, their hard function is not single-valued as a computation on oracles. Indeed, they explicitly emphasized that $\E^{\prAM} \not\subset \SIZE[\Omega(2^n/n)]$ remained an interesting open problem. (This question is no longer open anymore: the recent result of~\cite{CHR24, Li24} shows that $\S_2\E \not\subseteq \SIZE[2^n / n]$, and since $\S_2\E \subseteq \E^{\prAM}$~\cite{ChakaravarthyR11}, we have $\E^{\prAM} \not\subseteq \SIZE[2^n / n]$ as well.)%

Fixed-polynomial size lower bounds for \emph{polynomial time} with a $\prMA$ oracle have been known for many years. Santhanam~\cite{Santhanam09} proved that $\MA/_1$ and $\prMA$ require $n^k$-size circuits for all $k$; scaling up his result to $\MAEXP$ only implies a half-exponential size lower bound. In their work on improving Karp-Lipton collapses, Chakaravarthy and Roy~\cite{ChakaravarthyR11} prove that $\NP \subseteq \P/\poly$ implies $\PH=\P^{\prMA}$. This collapse immediately implies~\cite{Kannan82} that for every $k$, there is a function in $\P^{\prMA}$ not in $\SIZE[n^k]$. Their resulting $\P^{\prMA}$ machine asks non-promise queries (so it is not ``smart''), but the machine still gives consistent accept/reject answers regardless of how queries outside of the promise are answered. However, scaling their result up to $\E^{\prMA}$ also only yields a half-exponential size lower bound (see \autoref{appendix: E to prMA halfexp}). To circumvent these half-exponential barriers, we need a totally different approach.

\paragraph{A New Algorithm for Range Avoidance.} Like previous works~\cite{CHR24, Li24, CLL25}, our circuit lower bounds are proved by designing new algorithms for solving the \emph{Range Avoidance} ($\Avoid$) problem, introduced by \cite{KleinbergKMP21}. In this problem, we are given a circuit $C$ which has $n$ inputs and $n+1$ outputs, and the task is to find a $y$ of $n+1$ bits which is not in the range of $C$. This is a fundamental meta-complexity problem which defines a \emph{formal task} associated with counting arguments like the one described in the first paragraph of this paper: if $C$ takes in the description of a circuit of size $2^n/(10n)$ and outputs a $2^n$-bit truth table, then solving $\Avoid$ on $C$ directly corresponds to finding a near-maximum hard function. Range Avoidance has seen a flurry of work in recent years (for example~\cite{Korten21, DBLP:conf/approx/GuruswamiLW22,DBLP:conf/focs/RenSW22,DBLP:conf/stoc/IlangoLW23,DBLP:conf/stoc/ChenHLR23,DBLP:conf/stoc/ChenL24,DBLP:conf/focs/KortenP24} and Korten's recent survey~\cite{Korten25}). 

\begin{theorem}
    For every constant $d\ge 1$, there is a deterministic polynomial-time algorithm $\calA$ with access to a $\pr\MA$ oracle and one bit of advice $\{\alpha_n\}_{n\in\N}$ such that the following holds for infinitely many integers $n$. For every circuit $C: \{0, 1\}^n \to \{0, 1\}^{n+1}$ of size at most $n^d$, $\calA(C, \alpha_n)$ is smart (i.e., only makes queries to the $\pr\MA$ oracle that is indeed in the promise) and outputs a string $y \in \{0, 1\}^{n+1} \setminus \Range(C)$.
\end{theorem}

\begin{remark}
    As discussed in~\cite[Section 1.2.1]{CHR24}, to obtain circuit lower bounds, we need to design \emph{single-valued} algorithms for $\Avoid$. Roughly speaking, an algorithm $\calA$ for a search problem is \emph{single-valued}, if for each instance $C$, there is a canonical output $y_C$ that is outputted by every (or most) computational paths of $\calA$. We note that every smart $\FP^{\pr\MA}$ algorithm is necessarily single-valued, since the algorithm is deterministic and only makes queries inside the promise. (In contrast, if we drop the smartness requirement, then an $\FP^{\pr\MA}$ algorithm may not be single-valued, as the answer can depend on the values returned by the oracle on queries outside the $\pr\MA$ promise.)
\end{remark}

\subsection{Techniques} At a high level, we follow the ``iterative win-win'' method of~\cite{CLORS23} for pseudodeterministic constructions. Recall that a pseudodeterministic algorithm~\cite{DBLP:journals/eccc/GatG11} is a randomized algorithm that on an input $x$, with high probability outputs a \emph{canonical} solution $f(x)$ that does not depend on its internal randomness. (In other words, they are single-valued $\BPP$ or $\ZPP$ algorithms.)

The core ingredient underlying the iterative win-win method is a \emph{hardness-randomness tradeoff} with \emph{uniform reconstruction}. We view our main technical contribution as finding the right version of this ingredient that is tailored to the circuit lower bound we are proving.

\subsubsection{The Iterative Win-Win Method}
Let $\Pi = \{\Pi_n\}$ be a family of dense properties, i.e., $\Pi_n\subseteq \{0, 1\}^n$ and $|\Pi_n| \ge 2^n / \poly(n)$. (In~\cite{CLORS23}, $\Pi$ is the set of prime numbers, while in our paper, $\Pi$ is the set of hard truth tables.) Our goal is to design an efficient algorithm that successfully outputs an element of $\Pi_n$ for infinitely many input lengths $n$.

Consider two input lengths $n$ and $N$, where $n < N < \poly(n)$. Imagine that there is a deterministic ``brute force'' algorithm $\BF_n$ that successfully finds an element of $\Pi_n$ in time $T$. (The relationship between $T$ and $n$ may be arbitrary, since we start from the \emph{real} brute force algorithm where $T = \exp(n)$, and will reduce $T$ to be gradually closer to $\poly(n)$ as the win-win argument proceeds.) On such pair of input lengths, we perform a win-win argument of the following form:
\begin{description}
    \item [(Win)] If $\Pi_N$ is ``helpful'' for speeding up the brute force algorithm $\BF_n$, then we obtain an algorithm in $\poly(N)$ time that correctly simulates $\BF_n$. In~\cite{CLORS23}, this algorithm is randomized, which means that we obtain a randomized algorithm in $\poly(N)$ time that finds an element in $\Pi_n$. Moreover, since the randomized algorithm always outputs the output of $\BF_n$, it is also a \emph{pseudodeterministic} algorithm. In our work, this algorithm will run in $\smart\FP^\prMA$.
    \item [(Improve)] If $\Pi_N$ is ``not helpful'' for speeding up $\BF_n$, then we (somehow) obtain an algorithm $\calA_N$ for finding an element of $\Pi_N$ with time complexity $T' \le \poly(T)$. Since $T' \le 2^{O(n)} \ll 2^{O(N)}$, we make progress by bringing $T'$ closer to $\poly(N)$ than $2^{O(N)}$. We then proceed to the next level of the win-win argument by replacing $n$ with $N$, setting $\BF_N$ to be the ``better'' brute force algorithm $\calA_N$, and replacing $N$ with another input length $N' := \poly(N)$.
    
    Choosing the parameters carefully, we can guarantee that at some level of the win-win argument, we will have $T \le \poly(n)$ and we can afford to run $\BF_n$ (details omitted).
\end{description}

To summarize, for the sake of iterative win-win, we need that for every possible combination of $\BF_n$ and $\Pi_N$, either $\Pi_N$ is ``helpful'' for speeding up $\BF_n$, or $\BF_n$ implies an algorithm for finding an element in $\Pi_N$ in time $\poly(T) = \poly(|\BF_n|)$. This sounds familiar to a \emph{hardness-randomness tradeoff.}

\paragraph{The hardness-randomness tradeoff.} It is well-known~\cite{NisanW94, IW97} that circuit lower bounds imply derandomization. More precisely, letting $L \in \BPP$, there is a deterministic machine ${\rm IW97}(x, f)$ that takes an instance $x$ and a hard truth table $f$, runs in $\poly(|x|, |f|)$ time, and outputs $L(x)$. Treating the problem of finding an element in $\Pi_N$ as $L$ and (the computational history of) $\BF_n$ as $f$, we can see that:\footnote{Although ``finding an element in $\Pi_N$'' is technically not a (decisional) $\BPP$ problem, the hardness-randomness tradeoffs in~\cite{NisanW94, IW97} can still solve it by iterating through a complexity-theoretic PRG constructed from the hard truth table.}
\begin{itemize}
    \item either the circuit complexity of $\BF_n$ is small (for reasons that will be evident later, we will say that $\Pi_N$ is ``helpful'' for speeding up $\BF_n$ in this case);
    \item or $\BF_n$ implies a faster ($\poly(T)$-time) algorithm for constructing $\Pi_N$ by invoking ${\rm IW97}$.
\end{itemize}

We note that in our paper, since $\Pi_n$ is the set of hard truth tables, instead of \cite{NisanW94, IW97}, we need a hardness-randomness tradeoff for solving $\Avoid$. Looking ahead, we will use the \Jerabek--Korten reduction (\cite{Jerabek04, Korten21}, see also \autoref{thm: Jerabek-Korten}).

\paragraph{Uniform reconstruction.} We still need to justify that if the circuit complexity of $\BF_n$ is small, then $\Pi_N$ is helpful for speeding up $\BF_n$. At first glance, this seems to follow from the \emph{reconstruction argument} of~\cite{NisanW94}: Given a truth table $f$ and an adversary $\Pi_N$ breaking the PRG constructed from $f$, we can efficiently compute a small circuit for $f$. However, the reconstruction procedure of~\cite{NisanW94} requires some advice bits depending on $f$, thus we do not obtain a genuine speed-up algorithm for $\BF_n$. (This is also the same reason that \cite{NisanW94, IW97} requires a lower bound against \emph{non-uniform circuits} instead of a \emph{uniform probabilistic} ($\BPP)$ algorithm.)

For the task of constructing prime numbers, \cite{CLORS23} gets around this issue using the hardness-randomness tradeoff with \emph{uniform reconstruction} in~\cite{CT21b}. Applying a highly non-trivial arithmetization of $\BF_n$ in~\cite{gkr15}, the reconstruction algorithm can compute the needed advice bits ``on the fly''. We omit the details here, but the end result is that there is indeed a $\BPP$ algorithm (without any non-uniform advice) that computes (each output bit of) $\BF_n$ given oracle access to $\Pi_N$; in this sense, we say that $\Pi_N$ is ``helpful'' for speeding up $\BF_n$.

\subsubsection{Towards Solving Range Avoidance}
As discussed before, we use the \Jerabek--Korten reduction as the hardness-randomness tradeoff. In what follows, we let $T$ be the length of the ``hard truth table'' $f$, let $n$ be the size of the $\Avoid$ instance $G$ we are solving, and think of $T$ as an arbitrary parameter between $\poly(n)$ and $2^{O(n)}$.
\begin{theorem}[{Informal; \cite{Jerabek04, Korten21}}]\label{thm: JK informal}
    There is an algorithm $\JK(G, f)$ running in deterministic $\poly(n, T)$ time with access to an $\NP$ oracle such that, if the circuit complexity of $f$ is $\ge \poly(n, \log T)$, then $\JK(G, f)$ solves $\Avoid$ on the instance $G$.
\end{theorem}

This is indeed a hardness-randomness tradeoff: Given a truth table $f$ with large enough circuit complexity (and arbitrary length $T$), $\JK$ solves a derandomization task (namely $\Avoid$). However, as discussed before, we need a hardness-randomness tradeoff with \emph{uniform reconstruction}. That is, imagine that $\JK(G, f)$ \emph{fails} to solve $\Avoid$ on $G$, in what complexity class can we compute $f$ \emph{uniformly}?

\paragraph{A na\"ive attempt.} Suppose that $f$ is (some encoding of) the computational history of a $\TIME[T]^\NP$ computation. If $\JK(G, f)$ fails to solve $\Avoid$ on $G$, then by \autoref{thm: JK informal}, $f$ admits a small circuit. Our goal here is to compute $f$ in a uniform complexity class. Now, this sounds familiar to the \emph{Karp--Lipton theorem}~\cite{KarpL80}.

The Karp--Lipton theorems are a family of theorems showing that non-uniform collapses imply uniform collapses. For example:
\begin{quote}
	If $\EXP \subset \P/_\poly$ then $\EXP = \Sigma_2\P$ (\cite{KarpL80}, attributed to Albert Meyer).
\end{quote}

Such theorems are usually proved by exploiting small circuits which encodes the computational history of a complex computation. Indeed, the $\Sigma_2\P$ algorithm in the above theorem guesses a small circuit encoding the computational history of the $\EXP$ computation and uses a universal quantifier to verify that all of the exponentially many computation steps are executed correctly.

We apply a similar idea in our case. Every problem in $\TIME[T]^\NP$ reduces to finding the lexicographically smallest\footnote{We actually consider lexicographically \emph{largest} satisfying assignment in our formal proof, but it is easy to see that they are equivalent.} satisfying assignment of a size-$\poly(T)$ (satisfiable) $3$-CNF~\cite{Krentel88} and we can define the computational history $f$ to simply be this lex-first satisfying assignment. If $f$ has circuit complexity at most $s = \poly(n, \log T)$, then we can compute each bit of $f$ in $\Sigma_4\TIME[\poly(s)]$: use an existential quantifier to guess a size-$s$ circuit $C$ (whose truth table is supposedly $f$), use a universal quantifier to guess a size-$s$ circuit $C'$ (competing against $C$), use an existential quantifier to guess the smallest index $i$ such that $C(i) < C'(i)$, and finally use a universal quantifier to verify that for every $j < i$, $C(j) = C'(j)$. (We have omitted many details here such as verifying $C$ and $C'$ encode satisfying assignments.) It follows that we obtain a reconstruction algorithm in $\Sigma_4\P$, which implies a near-maximum circuit lower bound for $\Sigma_4\E$. Unfortunately, this is worse than the near-maximum circuit lower bound proved by simple diagonalization~\cite{Kannan82, MVW99}, which results in a hard function in no more than $\Sigma_3\E$.

We remark that by applying a suitable PCP~\cite{AroraS98, AroraLMSS98} on the $3$-CNF produced by~\cite{Krentel88} and using an efficient comparator for Reed--Muller codewords~\cite{Hirahara15}, it is possible to obtain an $\S_2\P$ reconstruction algorithm. However, this still only results in a near-maximum circuit lower bound for $\S_2\E$, which is no better than previous results~\cite{CHR24, Li24}.\footnote{In fact, this \emph{is} the proof of $\S_2\E/_1\not\subseteq\SIZE[2^n/n]$ presented in~\cite{CHR24}. The only difference is that instead of a general PCP, \cite{CHR24} applied a specific Reed--Muller encoding tailored to the $\JK$ procedure. This allowed~\cite{CHR24} to prove their lower bounds in a relativizing way.} %

\paragraph{The crucial insight.} We observe that the $\JK$ procedure is weaker than a full-fledged $\P^\NP$ algorithm. In particular, it only makes $r := \poly(n, \log T)$ many adaptive rounds of $\NP$ queries (see \autoref{thm: Jerabek-Korten}). It can be shown, using nearly the same argument as~\cite{Krentel88}, that any $\TIME[T]^\NP$ computation with only $r$ adaptive rounds of $\NP$ queries reduces to finding a satisfying assignment of a size-$\poly(T)$ $3$-CNF with the \emph{lexicographically smallest length-$r'$ prefix}, where $r' := r\log T$. Since $r' \le \poly(n, \log T)$ is small, our uniform reconstruction algorithm can afford to enumerate these $r'$ bits one by one.

In particular, assuming the computational history has circuit complexity at most $s$, the uniform reconstruction can be done in $\TIME[\poly(r', s)]^{\Sigma_2\P}$. We maintain the current prefix $\pfx$ which is initially empty; after $r'$ rounds, $\pfx$ should be equal to the lexicographically smallest length-$r'$ prefix of any satisfying assignment. One can test whether there is a satisfying assignment with (circuit complexity $\le s$ and) $\pfx$ as a prefix in $\Sigma_2\P$: use an existential quantifier to guess a size-$s$ circuit $C$ and use a universal quantifier to verify that $C$ is a satisfying assignment with prefix $\pfx$. In each round, if there is a satisfying assignment with ($\pfx$ concatenated with $0$) as a prefix, then we append $0$ to $\pfx$; otherwise we append $1$.

This gives a uniform reconstruction procedure in $\FP^{\Sigma_2\P}$. Applying a standard PCP theorem~\cite{AroraS98, AroraLMSS98} to the $3$-CNF, one can verify whether $C$ is a satisfying assignment by random coin tosses, which improves the complexity of the reconstruction procedure to $\FP^{\prMA}$. Moreover, a careful examination of our proof reveals that, on all good input lengths, the reconstruction only makes smart queries to the $\prMA$ oracle. (We can use a bit of advice to signal which input lengths are good.)

We develop the necessary machinery for manipulating $\P^\NP$ computations with a bounded number of adaptive rounds of $\NP$ queries, including a \cite{Krentel88}-style characterization, a PCP theorem, and the reconstruction argument. Note that we actually deal with the class $\P^\NP\RoundLength{r}{s}$, where each $\NP$ query has witness length at most $s$. The main reason to deal with the number of rounds and query lengths separately is that, using the standard search-to-decision reduction for $\SAT$, we need $s\cdot r$ adaptive rounds of $\NP$ queries to compute the ``computational history'' of $\P^\NP\RoundLength{r}{s}$ (see \autoref{lemma: computational history of P NP a}); fortunately this does not affect our overall proofs.

Finally, our techniques fall short of proving a near-maximum lower bound for smaller classes such as (the promise version of) $\MAEXP$. In \autoref{sec: conclusion}, we discuss the possibility of extending our results to these smaller classes.

\section{Preliminaries}

We use $x\circ y$ to denote the concatenation of two binary strings $x$ and $y$.

A function $f: \N \to \N$ is \emph{nice} if there is a deterministic algorithm that given $n$ as input, outputs the value $f(n)$ in time $\polylog(n, f(n))$. This is just a somewhat arbitrary definition of functions that are ``not too pathological''; most functions that appear as complexity bounds are nice, such as $f(n) = n^c$, $f(n) = 2^{cn}$, or $f(n) = 2^{n^c}$ for constants $c > 0$.

\subsection{Smart Reductions to Promise Problems} \label{sec: smart reductions}
A \emph{promise problem}~\cite{EvenSY84} is a pair of languages $(\Pi_\Yes, \Pi_\No)$ with $\Pi_\Yes \cap \Pi_\No = \varnothing$. An algorithm $\calA$ \emph{solves} the promise problem if $\calA$ accepts every input in $\Pi_\Yes$ and rejects every input in $\Pi_\No$. The algorithm is allowed to behave arbitrarily on inputs not in the promise ($x\not\in \Pi_\Yes\cup\Pi_\No$).

Let $\Pi = (\Pi_\Yes, \Pi_\No)$ be a promise problem and $L$ be a language. We say that there is a \emph{smart} reduction from $L$ to $\Pi$~\cite{GrollmannS88}, denoted as $L \in \smart\P^\Pi$, if there is a deterministic polynomial-time oracle algorithm $\calA$ with oracle access to $\Pi$ such that, for every input $x$, $\calA^\Pi(x)$ outputs the value of $L(x)$ and never makes any query outside the promise $\Pi_\Yes \cup \Pi_\No$.

We say that $L \in \smart\P^\Pi/_{a(n)}$ if the oracle algorithm $\calA$ can access an advice string $\alpha_n \in \{0, 1\}^{a(n)}$ that only depends on the input length $n$. The algorithm $\calA$ is required to be correct and only make smart queries when given the correct advice string; however, when given an incorrect advice string, $\calA$ is allowed to behave arbitrarily.

\subsection{Merlin-Arthur Classes}
We recall the definition of $\prMA$:
\begin{definition}
    A promise problem $\Pi = (\Pi_\Yes, \Pi_\No)$ is in $\pr\MA$ if there is a deterministic polynomial-time algorithm $\calA(x, w, r)$ such that $|w|, |r| \le \poly(|x|)$ and the following holds:
    \begin{itemize}
        \item {\bf Completeness:} For every input $x\in \Pi_\Yes$, there exists some $w$ such that
        \[\Pr_r[\calA(x, w, r)\text{ accepts}] = 1.\]
        \item {\bf Soundness:} For every input $x \in \Pi_\No$ and every $w$,
        \[\Pr_r[\calA(x, w, r)\text{ accepts}] \le 1/2.\]
    \end{itemize}
    Similarly, we say $\Pi$ is in $\pr\MA_\E$ (the \emph{exponential-time} analogue of $\pr\MA$) if $\calA$ is allowed to run in $2^{O(|x|)}$ time and $|w|, |r| \le 2^{O(|x|)}$.
\end{definition}

One can replace the soundness parameter $1/2$ with any small constant by repetition. One can also replace the completeness parameter with any constant larger than the soundness parameter (in our case $1/2$) without altering the definition of $\prMA$~\cite{ZachosF87, GoldreichZ11}.

When we consider oracle access to the class $\pr\MA$, it is without loss of generality to have the following specific $\pr\MA$-complete problem in mind: We are given a circuit $A(w, r)$ as input, and
\begin{align*}
    \Pi_\Yes :=&\, \{A: \exists w, \Pr_r[A(w, r) = 1] = 1\},\\
    \Pi_\No :=&\, \{A: \forall w, \Pr_r[A(w, r) = 1] \le 1/2\}.
\end{align*}

In \autoref{sec: conclusion}, we also need the definition of $(\MA\cap\coMA)/_{a(n)}$. Note that this is different from $\MA/_{a(n)} \cap \coMA/_{a(n)}$ in that there needs to be a \emph{single} machine with the same sequence of advice strings that decides the problem in $\MA$ and $\coMA$ simultaneously; see also Section 3.5 of the journal version of \cite{Chen25}.
\begin{definition}
    Let $a(n)$ be a nice function. A language $L$ is in $(\MA\cap\coMA)/_{a(n)}$ if there exists a deterministic polynomial-time algorithm $\calA(x, w, r, \alpha)$ and a sequence of advice strings $\{\alpha_n\}_{n\in \N}$ with each $|\alpha_n| = a(n)$, such that the following hold:
    \begin{itemize}
        \item $|w|, |r| \le \poly(|x|)$ and $\calA(x, w, r, \alpha)$ outputs a value in $\{0, 1, \bot\}$. Think of $w$ as Merlin's proof of the statement ``$L(x) = b$'' for some bit $b$, $r$ as Arthur's randomness, and $\alpha$ as the advice string; if $\calA$ accepts Merlin's proof then it outputs $b$, otherwise it outputs $\bot$.
        \item {\bf Completeness:} For every $x \in \{0, 1\}^\star$, there exists a string $w$ such that
        \[\Pr_r[\calA(x, w, r, \alpha_{|x|}) = L(x)] = 1.\]
        \item {\bf Soundness:} For every $x \in \{0, 1\}^\star$ and $w$,
        \[\Pr_r[\calA(x, w, r, \alpha_{|x|}) = 1 - L(x)] \le 1/2.\]
    \end{itemize}
\end{definition}
The exponential-time analogue of the above class, namely $(\MA_\E\cap\coMA_\E)/_{a(n)}$, is defined similarly.

\subsection{Self-Corrector for Low-Degree Polynomials}
Let $\calD$, $\calA$ be finite sets, we say that two functions $f, g: \calD \to \calA$ are \emph{$\delta$-close} if
\[\Pr_{x\gets \calD}[f(x) \ne g(x)] \le \delta.\]

We need a random self-correction procedure for low-degree polynomials:
\begin{theorem}[Low-Degree Self-Corrector~{\cite{GemmellS92, Sud95}}]\label{thm: low-degree correction}
    There is a probabilistic oracle algorithm $\Corr$ such that the following holds. Let $\F$ be a finite field, $m, \Delta \in \N$ such that $|\F| > \omega(\Delta^2)$. Let $P: \F^m \to \F$ be a polynomial of total degree at most $\Delta$, and $g: \F^m \to \F$ be any function that is $1/4$-close to $P$, then for all $\vec{x} \in \F^m$, $\Corr^g(\F, m, \Delta, \vec{x})$ runs in time $\poly(\Delta, \log|\F|, m)$ and outputs $P(\vec{x})$ with probability at least $2/3$. Moreover, if $g = P$, then for all $\vec{x} \in \F^m$, $\Corr^g(\F, m, \Delta, \vec{x})$ outputs $P(\vec{x})$ with probability $1$.
\end{theorem}

\section{Bounded-Adaptive Queries to an \texorpdfstring{$\NP$}{NP} Oracle}
Crucial to our near-maximum circuit lower bounds is the investigation of a particular subclass of $\P^\NP$, where we make \emph{multiple} rounds of $\NP$ queries with \emph{short} witnesses. Fix parameters $t(n), r(n), s(n)$, we say an algorithm runs in
\[\TIME[t(n)]^\NP\RoundLength{r(n)}{s(n)},\]
if it is a deterministic algorithm running in $t(n)$ time with access to an $\NP$ oracle, such that:
\begin{itemize}
    \item {\bf Round complexity:} The algorithm makes at most $r(n)$ rounds of $\NP$ queries. Each round can issue an unbounded number of queries in parallel, and the queries can depend on the answers of $\NP$ queries in previous rounds.
    \item {\bf Witness length:} Every $\NP$ query can be solved by a nondeterministic Turing machine guessing at most $s(n)$ nondeterministic bits. That is, they can be reduced to $\SAT$ instances with at most $s(n)$ many inputs (in deterministic $t(n)$ time).
\end{itemize}

Since $t(n)$ is an upper bound on the time complexity of a $\TIME[t(n)]^\NP\RoundLength{r(n)}{s(n)}$ machine, throughout this paper we will assume that $r(n), s(n) \le t(n)$. We also assume that the number of $\NP$ queries made in each round, as well as the sizes of the $\NP$ queries, are at most $t(n)$.

This class is important for two reasons:
\begin{enumerate}
    \item The \Jerabek--Korten reduction~\cite{Jerabek04, Korten21}, an important ingredient for our circuit lower bounds (and also for~\cite{CHR24, Li24}), runs in $\P^\NP\RoundLength{r(n)}{s(n)}$ for some modest parameters $r(n)$ and $s(n)$. It turns out that our main proof in \autoref{sec: near-max lower bounds} needs to keep track of $r(n)$ and $s(n)$ carefully. This complexity upper bound of the \Jerabek--Korten reduction is proven in~\autoref{sec: Jerabek-Korten}.
    \item Our main proof also needs an ``encoded computational history'' of this class, which is a slight adaptation of the PCP theorem~\cite{AroraS98, AroraLMSS98}. This class is capable of computing the encoded computational history of itself with a moderate overhead; in fact, the encoded computational history of $\P^\NP\RoundLength{r(n)}{s(n)}$ can be computed in $\P^\NP\RoundLength{r(n)\cdot s(n)}{s(n)}$. This is established in~\autoref{sec: encoded computational history}.
\end{enumerate}

\subsection{The \Jerabek--Korten Reduction}\label{sec: Jerabek-Korten}

A core ingredient in the previous near-maximum circuit lower bounds~\cite{CHR24, Li24} and our new lower bounds is the \Jerabek--Korten reduction from $\Avoid$ to constructing a hard truth table. This reduction first appeared in~\cite{Jerabek04} in the language of bounded arithmetic, and was rephrased in the language of total search problems and computational complexity in~\cite{Korten21}.

Given an $\Avoid$ instance $G: \{0, 1\}^n \to \{0, 1\}^{2n}$ and the truth table of a function $f$ with sufficiently high circuit complexity, the reduction $\JK(G, f)$ runs in deterministic $\poly(|G|, |f|)$ time with access to an $\NP$ oracle and outputs a string $y \in \{0, 1\}^{2n}\setminus\Range(G)$. In the following theorem, we show that the number of \emph{rounds of parallel} $\NP$ oracle queries in this reduction can be made $O(n + \log |f|)$, and the witness length of each $\NP$ query is at most $n$.

\begin{theorem}\label{thm: Jerabek-Korten}
    There is an algorithm $\JK$ that takes an $\Avoid$ instance $G: \{0, 1\}^n \to \{0, 1\}^{2n}$ (with $|G| \geq n$) and a truth table $f \in \{0, 1\}^T$ as inputs, and outputs either 
    \begin{itemize}
        \item the message \textnormal{\texttt{``easy''}} and a circuit $C$ of size $O(|G|\log T)$ that computes the truth table $f$, or
        \item the message \textnormal{\texttt{``hard''}} and a string $y \in \{0, 1\}^{2n}\setminus \Range(G)$.
    \end{itemize}
    The algorithm $\JK$ runs in $\TIME[\poly(T, |G|)]^\NP\RoundLength{O(n + \log(T/n))}{n}$.
\end{theorem}

\begin{proof} We shall follow the \Jerabek--Korten procedure~\cite{Jerabek04,Korten21} closely, and make some minor modifications to optimize the number of query rounds. First, if $T \leq n$, then the theorem is trivial and no oracle queries are needed: we can just output \texttt{``easy''} along with a circuit $C$ of size $O(T) \leq O(n)$ and $\log_2 T$ inputs, where $C$ has the truth table of $f$ hard-coded in a standard way. 

So assume $T > n$. Let $k = \lceil \log_2 (T/n) \rceil$ and let $T' = n 2^k$. Let $f' = f 0^{T'-T}$; that is, $f'$ is $f$ padded with $T' - T$ zeroes. Note that $|f'| < 2|f|$.

Our algorithm proceeds in $k$ stages, starting with stage $k$ and decreasing down to stage $1$. Initially we set $y_k = f'$, a string of $n \cdot 2^k$ bits. 

At the beginning of stage $i$, the current string $y_i$ under consideration has length $n \cdot 2^i$. We partition $y_i$ into contiguous $2n$-bit blocks $z_1,\ldots,z_{2^{i-1}}$, and form $\SAT$ queries $\phi_1,\ldots,\phi_{2^{i-1}}$, where $\phi_j$ is satisfiable if and only if there is an $x_j$ such that $G(x_j) = z_j$. We are encoding a nondeterministic computation that takes at most $|G| \cdot \polylog |G|$ time, guesses $x_j$ and evaluates the circuit $G$. Using the standard translation of nondeterministic time $t(n)$ into SAT queries of length $t(n) \cdot \polylog( t(n))$ (which holds for all reasonable models of computation~\cite{Schnorr78,Cook88,GurevichS89}) we can construct each $\phi_j$ so that $|\phi_j| \leq |G|\cdot \polylog(|G|)$ for all $j$. 

If $\SAT$ answers \textsc{No} on some $\phi_j$, then we immediately return \texttt{``hard''} and the corresponding string $z_j$ as a solution to $\Avoid$. Otherwise, for every $\Phi_j$, let $x_j \in \{0, 1\}^n$ denote its lexicographically smallest satisfying assignment. The simplest way to compute $x_j$ would cost $n$ parallel rounds per stage: one round for each variable value. Letting $b \geq 1$ be a parameter, we observe that one can make only $m := \lceil n/b \rceil$ extra rounds per stage in a way that multiplies the overall running time by an $O(2^b)$ factor. For simplicity, let us round up the number of variables in each $\phi_j$ to be exactly $b m \leq 2n$ by adding some dummy variables (later, we will just remove them so that the total number is $n$ again). Suppose for some $\ell \in \{0,\ldots,m-1\}$ and for every $j=1,\ldots,2^{i-1}$, the $(\ell\cdot b)$-bit prefix $p_j$ of $x_j$ has already been determined. Then in the next round, we form a collection of $2^b \cdot 2^{i-1}$ parallel $\SAT$ queries $\psi_{q,\ell,j}$ over all $q \in \{0,1\}^b$, such that $\psi_{q,\ell,j}$ is satisfiable if and only if there is an $x' \in \{0, 1\}^n$ such that $G(x') = z_j$, the first $\ell \cdot b$ bits of $x'$ equals $p_j$, and the next $b$ bits of the prefix of $x'$ equals $q$. Similarly to the $\phi_j$, we can encode each $\psi_{q,\ell,j}$ so that they each have size at most $|G|\cdot \polylog(|G|)$ as well. The yes/no answers from all $2^b \cdot 2^{i-1}$ queries $\psi_{q,\ell,j}$ directly determine the next $b$ bits of each $x_j$, and we update each prefix $p_j$ accordingly by appending the lexicographically first $q$ such that $\psi_{q,\ell,j}$ is satisfiable. Observe that the running time in each stage is multiplied by a factor of $2^b$.

After we have determined all $n$ bits of all $x_j$, we form a new string $y_{i-1} = x_1 \cdots x_{2^{i-1}}$ which has length $n \cdot 2^{i-1}$, which we pass to stage $i-1$, repeating the above process.

When we reach stage $1$, the resulting $y_0$ has length $n$. In this case, we return \texttt{``easy''}, since we can (by a standard argument, see \cite{Korten21,GoldreichGM86}) construct a circuit of size $O(|G| \cdot \log T)$ whose truth table is exactly $f'$, and therefore we can construct a circuit for $f$ as well. 

As we take $t = \lceil n/b \rceil + 1$ rounds in each stage, the total number of parallel rounds is $O(t \cdot k)$. Letting $b = \lceil \log(T\cdot |G|) \rceil$, we multiply the running time by $T\cdot |G|$ (a polynomial factor) and the number of parallel rounds becomes $O((\lceil n/b \rceil + 1)k) \leq O((\lceil n/\lceil \log(T\cdot |G|) \rceil + 1\rceil) \log(T/n))  \leq O(n + \log(T/n))$. It is easy to see that every $\SAT$ query has witness length at most $n$.
\end{proof}

\subsection{Encoded Computational History}\label{sec: encoded computational history}
We consider the \emph{encoded computational history} of $\P^\NP\RoundLength{r(n)}{s(n)}$ machines. 

The computational history of an $\NP$ machine on an input $x$ is straightforward to define: we can simply provide an accepting computation path of the machine (a list of transitions taken by the machine from the initial configuration to the accept state). Applying the PCP theorem, we obtain an \emph{encoded} computational history of the $\NP$ machine that we can check efficiently with randomness.

For a $\P^{\NP}$ machine, we can reasonably define the computation history on an input $x$ to be the history of the $\P$ machine's transitions along with a record of the queries asked. For each query that answers \textsc{yes}, we also provide accepting computation paths from the $\NP$ machine in the history just after the query is made; for each \textsc{no} query, we insert an all-zero string. Observe that from this definition, finding the lexicographically largest accepting computation history determines the correct query answers (because we will always prefer an accepting computation history over an all-zeros string). Similarly, we obtain the \emph{encoded} computational history by applying a PCP theorem to the raw computational history, and this facilitates fast lexicographic comparisons~\cite{Hirahara15}. 

In what follows, we generalize the above observation to $\P^\NP\RoundLength{r(n)}{s(n)}$. We also define its \emph{encoded} computational history by applying the PCP theorem.

\def\prefix{\normalfont{\mathsf{pfx}}}

\begin{lemma}\label{lemma: computational history of P NP a}
    Let $M$ be a $\TIME[t(n)]^\NP\RoundLength{r(n)}{s(n)}$ machine that outputs $\ell(n)$ bits. For every input $x$, there is a circuit $C_x$ computable in deterministic $\poly(t(n))$ time given $x$, such that the following holds.
    
    Let $\alpha \in \{0, 1\}^{\poly(t(n))}$ denote the lexicographically largest satisfying assignment of $C_x$. Then $\alpha$ can be computed in $\TIME[\poly(t(n))]^\NP\RoundLength{s(n)\cdot r(n)}{s(n)}$. Moreover, letting $r'(n) := r(n)\lceil \log (t(n)+1)\rceil$, for every $\alpha' \in \{0, 1\}^{\poly(t(n))}$ such that $C_x(\alpha') = 1$ and the first $r'(n)$ bits of $\alpha$ and $\alpha'$ are equal, the substring of $\alpha$ from the $(r'(n)+1)$-th bit to the $(r'(n) + \ell(n))$-th bit is equal to the output of $M(x)$.
\end{lemma}

\def\hist{\mathsf{hist}}

\begin{proof}
    Let $C_x$ denote the following circuit. Its input consists of $\hist = (y_1, \dots, y_{r(n)}, out, z_{1, 1}, \dots, z_{r(n), t(n)})$ (in this order), with the following intended meanings:
    \begin{itemize}
        \item for each $i\in [r(n)]$, $y_i \in [0, t(n)]$ is the number of $\NP$ queries made in the $i$-th round of $M(x)$ that are accepted;
        \item $out \in \{0, 1\}^{\ell(n)}$ is the output of $M(x)$; and
        \item for each $i\in [r(n)]$ and $j\in [t(n)]$, $z_{i, j} \in \{0, 1\}^s$ is an $\NP$ witness for the $j$-th query made in the $i$-th round of $M(x)$.
    \end{itemize}
    We call such a tuple a ``computational history'' of $M(x)$; as we will demonstrate later, the lexicographically largest $\hist$ accepted by $C_x$ is the ``correct'' computational history of $M(x)$.

    Given any string $\hist$, we can simulate the computation of $M(x)$ according to $\hist$ in $t(n)$ time: Let $q_{i,j}(\hist)$ denote the $j$-th $\NP$ query made in the $i$-th round, then we answer \textsc{Yes} to this query if the entry $z_{i,j}$ (in $\hist$) satisfies $q_{i,j}(\hist)$, and we answer \textsc{No} otherwise. We stress that $q_{i,j}$ depends on the witnesses in previous rounds (i.e., $\{z_{i', j'}\}_{i' < i}$). We define $C_x(\hist)$ accepts if and only if in the above simulation:
    \begin{itemize}
        \item for each $i\in [r(n)]$, the number of accepted $\NP$ queries in the $i$-th round is exactly $y_i$; and
        \item the final output of $M(x)$ is equal to $out$.
    \end{itemize}

    Let $\hist^\star = (y_1^\star, \dots, y_{r(n)}^\star, out^\star, z_{1, 1}^\star, \dots, z_{r(n), t(n)}^\star)$ denote the ``correct'' computational history of $M(x)$. More precisely:
    \begin{itemize}
        \item for each $i\in [r(n)]$, $y_i^\star$ is the number of $\NP$ queries in the $i$-th round of $M(x)$ that are accepted;
        \item $out^\star$ is the final output of $M(x)$; and
        \item for each $i\in [r(n)]$ and $j\in [t(n)]$, if the $j$-th $\NP$ query in the $i$-th round of $M(x)$ is accepted, then $z_{i, j}^\star$ is the lexicographically largest witness of this query; otherwise $z_{i, j}^\star = 1^s$.
    \end{itemize}

    It is not hard to see that:
    \begin{claim}\label{claim: histstar is the correct comp history}
        $\hist^\star$ is the lexicographically largest satisfying assignment of $C_x$.
    \end{claim}
    \begin{claimproof}
        Clearly, $C_x(\hist^\star) = 1$. Assume that $\hist' = (y_1', \dots, y_{r(n)}', out', z'_{1, 1}, \dots, z'_{r(n), t(n)})$ is any string lexicographically larger than $\hist^\star$ such that $C_x(\hist') = 1$, we now derive a contradiction.

        Let $i\in [r(n)]$ and $j\in [t(n)]$. It is easy to show using induction that if $y_{i'}' = y_{i'}^\star$ for every $i' < i$, then $q_{i,j}(\hist') = q_{i,j}(\hist^\star)$, i.e., the two simulations of $M(x)$ on $\hist^\star$ and $\hist'$ makes the same $\NP$ query in the $i$-th round. The base case (where $i = 1$) is trivial. When $i > 1$, since $y_i^\star$ is the number of satisfiable $\NP$ queries among $\{q_{i-1,j}(\hist^\star)\}_{j\in [t(n)]} = \{q_{i-1,j}(\hist')\}_{j\in [t(n)]}$, it must be the case that the answers to every $q_{i-1,j}(\hist')$ are the same as those in $q_{i-1,j}(\hist^\star)$. Therefore, the $\NP$ queries made by $M(x)$ in the $i$-th round are the same when simulated using $\hist^\star$ or using $\hist'$.
        
        Suppose that there is some $i\in [r(n)]$ such that $y_i' \ne y_i^\star$. Let $i$ be the smallest such index, then for every $j < i$, we have $y_j' = y_j^\star$. By the induction argument above, we have $q_{i,j}(\hist') = q_{i, j}(\hist^\star)$ for every $j\in [t(n)]$. However, there are only $y_i^\star < y_i'$ many satisfiable queries in the $i$-th round, therefore the number of $j\in [t(n)]$ such that $z_{i,j}'$ satisfies $q_{i,j}(\hist')$ cannot be equal to $y_i'$. This contradicts our assumption that $C_x(\hist')$ accepts.

        It follows that $y_i' = y_i^\star$ for every $i\in [r(n)]$. The above induction implies that the $\NP$ queries and answers of $M(x)$ are exactly the same, when simulated using $\hist'$ or using $\hist^\star$. It follows that the output of $M(x)$ on $\hist'$ is exactly $out$, hence if $out' \ne out$ then $C_x(\hist') = 0$.
        
        Finally, recall that $z_{i, j}^\star$ is the lexicographically largest witness of $q_{i,j}(\hist^\star)$ (when satisfiable) or $1^s$ (when unsatisfiable). Therefore if some $z_{i, j}' > z_{i, j}^\star$, it must be the case that $q_{i,j}(\hist')$ is satisfiable but $z_{i, j}'$ does not satisfy it. Hence the number of satisfied queries in the $i$-th round cannot be equal to $y_i' = y_i^\star$, a contradiction to our assumption that $C_x(\hist')$ accepts.
    \end{claimproof}

    That is, $\alpha := \hist^\star$ is the lexicographically largest satisfying assignment of $C_x$. We can compute $\hist^\star$ by simulating $M(x)$ using the $\NP$ oracle. Note that we also need to find the lexicographically largest satisfying assignments of each $\NP$ query (i.e., $z_{i, j}^\star$). This can be done by a binary search procedure with a blow-up of the round complexity by a multiplicative factor of $s(n)$.

    \begin{algorithm}[H]
        \caption{Computing $\hist^\star$ in $\TIME[\poly(t(n))]^\NP\RoundLength{r(n)\cdot s(n)}{s(n)}$}
        \begin{algorithmic}[1]
            \For {$i\gets 1$ to $r(n)$}
                \State {Let $q_{i, 1}, \dots, q_{i, t(n)}$ be the $\NP$ queries issued in the $i$-th round}
                \State {$z_{i, j} \gets \varepsilon$ for every $j\in [t(n)]$}\Comment{$\varepsilon$ is the empty string (of length $0$)}
                \For {$k\gets 1$ to $s(n)$}
                    \For {$j\gets 1$ to $t(n)$}
                    
                    \Comment{We can issue the following $t(n)$ $\NP$ queries in parallel; each query has witness length at most $s(n)$}
                        \If {there is a satisfying assignment of $q_{i, j}$ with prefix $(z_{i, j}\circ 1)$}
                            \State {$z_{i,j}\gets z_{i,j}\circ 1$}
                        \Else
                            \State {$z_{i,j}\gets z_{i,j}\circ 0$}
                        \EndIf
                    \EndFor
                \EndFor
                \For {$j\gets 1$ to $t(n)$}
                    \State {Each $q_{i,j}$ is answered according to whether $z_{i,j}$ satisfies it}
                    \If {$q_{i,j}$ is not satisfied}
                        \State{$z_{i,j}\gets 1^s$}
                    \EndIf
                \EndFor
            \EndFor
        \end{algorithmic}
    \end{algorithm}

    Finally, the length of $(y_1, \dots, y_{r(n)})$ is $r'(n) \le r(n)\lceil \log (t(n)+1)\rceil$. In the proof of \autoref{claim: histstar is the correct comp history} we have already seen that for any $\hist'$ that agrees with $\hist^\star$ on the first $r'(n)$ bits (i.e., agrees on $y_1, \dots, y_{r(n)}$), if $C_x(\hist') = 1$, then the output part of $\hist'$ should be equal to $out^\star$. This establishes the ``Moreover'' part of the lemma.
\end{proof}

\paragraph{A PCP theorem.} We need a PCP theorem where the PCP \emph{contains an encoded version of the raw witness as a prefix}; this PCP theorem will be applied to the circuit $C_x$ as defined in \autoref{lemma: computational history of P NP a}. Fix an error-correcting code (looking ahead, we will use the Reed--Muller code). In our discussion, every PCP proof $\Pi = (\tilde{w}, \Pi')$ consists of two parts $\tilde{w}$ and $\Pi'$, where $\tilde{w}$ is an encoding of a witness $w$ and $\Pi'$ denotes auxiliary proof oracles.

Roughly speaking, we need the following completeness and soundness guarantees:
\begin{itemize}
    \item {\bf Completeness:} Every valid witness $w$ corresponds to a PCP proof $(\tilde{w}, \Pi')$ accepted with probability $1$, where $\tilde{w}$ is the encoding of $w$. 
    \item {\bf Soundness:} If $\tilde{w}$ is far from encodings of valid witnesses, then for every $\Pi'$, the PCP proof $(\tilde{w}, \Pi')$ is rejected with noticeable probability.
\end{itemize}

Such guarantees are satisfied by the classical Reed--Muller based PCP constructions~\cite{AroraS98, AroraLMSS98}. A nice exposition of these PCPs can be found in Harsha's PhD thesis~\cite{harsha2004robust} and we will, in particular, verify that the robust PCP for $\SAT$ (\textsc{robust-PCP-Circuit-SAT}) in~\cite[Section 5.4.1]{harsha2004robust} satisfies the above guarantees.

To be more precise, let us consider PCPs for the language $\CircuitSAT$. Let $C$ be a circuit of size $n$, we say that $w \in \{0, 1\}^n$ is a \emph{valid witness} for ``$C \in \CircuitSAT$'' if $w$ consists of a satisfying assignment of $C$ together with the values of intermediate gates of $C$ on this assignment. Set $m := \mleft\lceil\frac{\log n}{\log\log n}\mright\rceil$, $\F$ be a finite field of size at least $\alpha\cdot (mn^{1/m})^3 \le \polylog(n)$ where $\alpha \ge 1$ is a large enough universal constant, and $H \subseteq \F$ be a subset of size $n^{1/m}$ that contains $\{0, 1\}$. (Here we assume, without loss of generality, that $n^{1/m}$ is an integer.) Fix a bijection between $[n]$ and $H^m$, then every string $w \in \{0, 1\}^n$ can be seen as a function $w: H^m \to \{0, 1\}$, hence extends to a polynomial $\tilde{w}: \F^m \to \F$ with individual degree at most $n^{1/m}$. We call $\tilde{w}$ the \emph{Reed--Muller encoding} of $w$.

\begin{theorem}\label{thm: PCP for SAT with witness prefix}
    For every small enough constant $\delta > 0$, there exists a PCP verifier $V$ for $\CircuitSAT$ such that for every circuit $C$ of size $n$:
    \begin{itemize}
        \item {\bf Completeness:} For every valid witness $w$ for ``$C \in \CircuitSAT$'', there exists a PCP proof $\Pi = (\tilde{w}, \Pi')$ such that $\tilde{w}$ is the Reed--Muller encoding of $w$ and
        \[\Pr_z[V^\Pi(C, z)\text{ accepts}] = 1.\]
        Moreover, such a proof $\Pi$ can be computed in deterministic $\poly(n)$ time given $C$ and $w$ as inputs.
        \item {\bf Soundness:} For every PCP proof $\Pi = (\tilde{w}, \Pi')$, if $\tilde{w}$ is $\delta$-far from any polynomial $\tilde{w}':\F^m \to \F$ of total degree at most $m\cdot |H|$ such that $\tilde{w}'|_{H^m}$ is a valid witness for ``$C \in \CircuitSAT$'', then 
        \[\Pr_z[V^\Pi(C, z)\text{ accepts}] \le 1/2.\]
        \item {\bf Complexity:} The PCP proof has length $|\Pi|\le \poly(n)$. The verifier $V$ takes a random string of length $|z| = O(\log n)$ and makes $\polylog(n)$ queries to $\Pi$.
    \end{itemize}
\end{theorem}
\def\ZoS{\mathsf{ZoS}}
\begin{proof}[Proof Sketch]
    The PCP in \cite[Section 5]{harsha2004robust} consists of functions $\tilde{A}: \F^m \to \F$, $P: \F^{3m+3} \to \F$, and $\Pi_\ZoS: \F^{3m+3} \to \F^{6m+6}$, where $\tilde{A}$ is supposed to be the Reed--Muller encoding of a valid witness, $P$ is an auxiliary proof oracle for \textsc{robust-PCP-Circuit-SAT}, and $\Pi_\ZoS$ is an auxiliary proof oracle for \textsc{Zero-on-Subcube}. We denote $\tilde{w} = \tilde{A}$ and $\Pi' = (P, \Pi_\ZoS)$.
    \begin{itemize}
        \item {\bf Completeness:} The discussion before \cite[Proposition 5.4.1]{harsha2004robust} shows how to compute $\tilde{A}$ and $P$ in polynomial time such that $\tilde{A}|_{H^m} = w$. Moreover, $P$ is a Yes instance for \textsc{Zero-on-Subcube} if and only if $w$ is a valid witness for ``$C \in \CircuitSAT$''. Then, \cite[Section 5.3.2]{harsha2004robust} shows how to compute $\Pi_{\ZoS}$ in polynomial time given $P$.
        \item {\bf Soundness:} This follows from \cite[Lemma 5.4.4]{harsha2004robust}. Note that $|\F| \ge \alpha\cdot (mn^{1/m})^3$ for some large enough constant $\alpha$, and $\delta$ is a small enough constant, hence the hypotheses of \cite[Lemma 5.4.4]{harsha2004robust} are satisfied and this lemma is applicable. The soundness parameter can be amplified to $1/2$ by repeating the verification procedure $O(\delta^{-1})$ times.
        \item {\bf Complexity:} As calculated after \cite[Remark 5.4.3]{harsha2004robust}, the query complexity is $O(m|\F|\log|\F|) \le \polylog(n)$ and the randomness complexity is $O(m\log |\F|) \le O(m\log (mn^{1/m})) \le O(\log n)$. The length of the PCP proof is $\poly(|\F|^m) \le \poly(n)$.\qedhere
    \end{itemize}

\end{proof}

In what follows, we will deal with circuits taking $t'(n) \le \poly(t(n))$ inputs. Therefore, we need to substitute $n$ by $t'(n)$ in the instantiation of Reed--Muller codes. This in particular means that $m := \lceil \frac{\log t'(n)}{\log\log t'(n)}\rceil$, $\F$ is a finite field of size $\Theta(m\cdot t'(n)^{1/m})^3$, and $H\subseteq \F$ is a subset of size $t'(n)^{1/m}$.

\paragraph{The ``encoded history''.} Let $M$ be $\TIME[t(n)]^\NP\RoundLength{r(n)}{s(n)}$ machine, $x \in \{0, 1\}^n$, and let $C_x$ denote the circuit guaranteed in \autoref{lemma: computational history of P NP a}. Let $V$ be a PCP verifier for $\CircuitSAT$ as described in \autoref{thm: PCP for SAT with witness prefix}, then $V$ accepts PCP proofs of the form $\Pi = (\tilde{w}, \Pi')$ where $\tilde{w}$ is the Reed--Muller encoding of a valid witness for $C_x$ and $\Pi'$ denotes auxiliary proof oracles. Note that the Reed--Muller code is systematic (i.e., the original data appears as a prefix of the encoding), hence $\tilde{w}$ contains a satisfying assignment of $C_x$ as a prefix. This satisfying assignment, in turn (by \autoref{lemma: computational history of P NP a}), contains $(\pfx, out)$ as a prefix, where $\pfx \in \{0, 1\}^{r'(n)}$ and $out$ is (purportedly) the output of $M(x)$.

\begin{theorem}\label{thm: PCP}
    Let $M$ be a $\TIME[t(n)]^\NP\RoundLength{r(n)}{s(n)}$ machine that outputs $\ell(n)$ bits, and $r'(n) := r(n)\cdot \lceil\log (t(n)+1)\rceil$. For any small enough constant $\delta > 0$, there is a ``PCP verifier'' $V^\Pi(x, z)$ such that for every $x \in \{0, 1\}^n$, there exists a string $\prefix_x \in \{0, 1\}^{r'(n)}$ such that the following holds:
    \begin{itemize}
        \item {\bf Completeness:} There exists a PCP proof $\Pi = (\tilde{w}, \Pi')$ such that $\Pr_z[V^\Pi(x, z)\text{ accepts}] = 1$, $\prefix_x$ is a prefix of $\Pi$, and $\tilde{w}$ is a polynomial of total degree $\le m|H|$. Such a proof can be computed in
        \[\TIME[\poly(t(n))]^\NP\RoundLength{r(n)\cdot s(n)}{s(n)}.\]
        \item {\bf Soundness:} Let $\Pi = (\tilde{w}, \Pi')$ be any PCP proof such that $\Pr_z[V^\Pi(x, z)\text{ accepts}] > 1/2$. Then $\tilde{w}$ is $\delta$-close to some polynomial $\tilde{w}': \F^m \to \F$ of total degree $m|H|$ such that the first $r'(n)$ bits of $\tilde{w}'$ are $\le \pfx_x$ (in lexicographic order). Moreover, if these first $r'(n)$ bits are exactly equal to $\pfx_x$, then the $(r'(n)+1 \sim r'(n)+\ell(n))$-th bits of $\tilde{w}'$ are equal to the output of $M(x)$.
        \item {\bf Complexity:} The PCP proof has length $|\Pi| \le \poly(t(n))$. The verifier $V$ takes a random string of length $|z| = O(\log t(n))$ and makes $\polylog(t(n))$ queries to $\Pi$.
    \end{itemize}
\end{theorem}
\begin{proof}
    Let $V^\Pi(x, z)$ be the verifier that given an input $x$, computes the circuit $C_x$ as in \autoref{lemma: computational history of P NP a} and runs $V_1^\Pi(C_x, z)$, where $V_1$ is the PCP verifier in \autoref{thm: PCP for SAT with witness prefix}. For each input $x$, let $\alpha_x$ denote the lexicographically largest satisfying assignment of $C_x$, and let $\pfx_x$ denote the length-$r'(n)$ prefix of $\alpha_x$.

    \begin{itemize}
        \item {\bf Completeness:} Let $w$ be the witness for ``$C_x \in \CircuitSAT$'' induced by the assignment $\alpha_x$. Then, by the {\bf Completeness} bullet in \autoref{thm: PCP for SAT with witness prefix}, there exists a PCP proof $\Pi = (\tilde{w}, \Pi')$ such that $\tilde{w}$ is the Reed--Muller encoding of $w$ and
        \[\Pr_z[V^\Pi(x, z)\text{ accepts}] = \Pr_z[V_1^\Pi(C_x, z)\text{ accepts}] = 1.\]
        Moreover, $\Pi$ can be computed in deterministic $\poly(t(n))$ time given $w$ (and $C_x$). Since $w$ can be computed in
        \[\TIME[\poly(t(n))]^\NP\RoundLength{s(n)\cdot r(n)}{s(n)},\]
        given $x$, $\Pi$ can be computed within asymptotically the same resource bound.

        \item {\bf Soundness:} Since $\Pr_z[V^\Pi(x, z)\text{ accepts}] > 1/2$, by the {\bf Soundness} bullet in \autoref{thm: PCP for SAT with witness prefix}, $\tilde{w}$ is $\delta$-close to some polynomial $\tilde{w}': \F^m \to \F$ of total degree at most $m\cdot |H|$ such that $\tilde{w}'|_{H^m}$ is a valid witness for ``$C_x \in \CircuitSAT$''. Since $\alpha_x$ is the lexicographically largest such witness, $\tilde{w}'|_{H^m}$ is lexicographically $\le \alpha_x$. Moreover, by \autoref{lemma: computational history of P NP a}, if the length-$r'(n)$ prefix of $\tilde{w}'$ coincides with $\pfx_x$, then the $(r'(n) + 1 \sim r'(n) + \ell(n))$-th bits of $\tilde{w}'$ is equal to the output of $M(x)$.
        \item {\bf Complexity:} Let $t'(n) \le \poly(t(n))$ denote the size of $C_x$. The {\bf Complexity} bullet in \autoref{thm: PCP for SAT with witness prefix} implies that $V$ accepts a proof of length $\poly(t'(n)) \le \poly(t(n))$, takes a random string of length $O(\log t'(n)) \le O(\log t(n))$, and makes $\polylog(t'(n)) \le \polylog(t(n))$ queries to $\Pi$.\qedhere
    \end{itemize}
\end{proof}

\section{A Near-Maximum Circuit Lower Bound}\label{sec: near-max lower bounds}
In this section, we prove \autoref{thm: main}.

\subsection{Instance-Wise Hardness-Randomness Tradeoff for \texorpdfstring{$\Avoid$}{Avoid}}

We start with the following theorem, which is essentially an \emph{instance-wise hardness-randomness tradeoff} for solving $\Avoid$ in $\P^\NP\RoundLength{r(n)}{s(n)}$ for some modest parameters $r(n)$ and $s(n)$. (See also \autoref{remark: instance-wise HvR}). Crucially, this hardness-randomness tradeoff has a \emph{uniform reconstruction procedure} computable in $\smart\P^\prMA$.

\begin{theorem}\label{thm: instance-wise HvR}
    Let $r(n), s(n) \le t(n)$ be nice functions and let $f: \{0, 1\}^n \to \{0, 1\}^{t(n)}$ be a multi-output function computable in $\TIME[t(n)]^\NP\RoundLength{r(n)}{s(n)}$. For every nice function $k(n) \le t(n)$, there are algorithms $\Solve_f$ and $\Recon_f$, which can be efficiently computed from the description of a $\SAT$-oracle machine for $f$, such that the following holds:
    \begin{itemize}
        \item $\Solve_f(x, G)$ takes $x \in \{0, 1\}^n$ and an $\Avoid$ instance $G: \{0, 1\}^{k(n)} \to \{0, 1\}^{2k(n)}$ of size at most $k(n)^4$ as inputs, outputs a string $y \in \{0, 1\}^{2k(n)}$, and runs in
        \[\TIME[\poly(t(n))]^\NP\RoundLength{r^\Solve(n) := r(n)s(n) + O(k(n) + \log t(n))}{s^\Solve(n) := \max\{s(n), k(n)\}}.\]
        \item $\Recon_f(x, i)$ takes $x \in \{0, 1\}^n$ and $i \in [t(n)]$ as  input, runs in deterministic $\poly(k(n), r(n), \log t(n))$ time with access to a $\prMA$ oracle $\Pi_f$, and outputs a bit.
        
        (As a minor remark, $\Recon_f(x, i)$ does not need access to the bad $\Avoid$ instance $G$.)
        
        \item Let $x \in \{0, 1\}^n$ and $G$ be an $\Avoid$ instance such that $\Solve_f(x, G)$ fails to output a string outside $\Range(G)$. Then for every $i \in [t(n)]$, the output of $\Recon_f(x, i)$ is the $i$-th bit of $f(x)$, and $\Recon_f(x, i)$ only makes smart $\pr\MA$ oracle queries.
    \end{itemize}
\end{theorem}
\begin{proof}
    We first describe the procedure $\Solve_f(x, G)$. Let $r'(n) := O(r(n)\log t(n))$, we first apply the PCP theorem (\autoref{thm: PCP}) on the function $f$ and obtain a ``PCP verifier'' $V^\Pi(x, z)$. For every input $x \in \{0, 1\}^n$, we also obtain a PCP proof $\Pi_x\in \{0, 1\}^{\poly(t(n))}$ as in the {\bf Completeness} item of \autoref{thm: PCP}. Then, $\Solve_f(x, G)$ invokes $\JK(G, \Pi_x)$ (\autoref{thm: Jerabek-Korten}), using $\Pi_x$ as the ``hard truth table'' to solve the $\Avoid$ instance $G$. We say that $\Solve_f(x, G)$ \emph{succeeds} if $\JK(G, \Pi_x)$ outputs the message \texttt{``hard''} and a non-output of $G$.
    
    We can compute $\Pi_x$ in $\TIME[\poly(t(n))]^\NP\RoundLength{r(n)\cdot s(n)}{s(n)}$. Then, $\JK(G, \Pi_x)$ runs in $\TIME[\poly(t(n), |G|)]^\NP\RoundLength{O(k(n)+\log t(n))}{k(n)}$. Since $|G| \le k(n)^4\le\poly(t(n))$, $\Solve_f(x, G)$ can be computed in
    \[\TIME[\poly(t(n))]^\NP\RoundLength{r(n)\cdot s(n)+O(k(n)+\log t(n))}{\max\{s(n), k(n)\}}.\]
    Next we describe the procedure $\Recon_f(x, i)$; assuming $\Solve_f(x, G)$ does not succeed, the goal of this procedure is to output the $i$-th bit of $f(x)$ in $\poly(k(n), r(n), \log t(n))$ time with smart access to a $\prMA$ oracle. For a PCP proof $\Pi = (\tilde{w}, \Pi')$, let $\tilde{w}': \F^m \to \F$ denote the polynomial of total degree at most $m|H|$ that is closest to $\tilde{w}$, and define $\widetilde{\pfx}(\Pi)$ and $\widetilde{out}(\Pi)$ to be the first $r'(n)$ bits and the $(r'(n) + 1 \sim r'(n) + t(n))$-th bits of $\tilde{w}'$, respectively. (Recall the parameters here: $m = O\mleft(\frac{\log t(n)}{\log\log t(n)}\mright)$, $|H| := \poly(t(n)^{1/m}) \le \polylog(t(n))$ and $|\F| = \Theta((m|H|)^3)$.) We also say that $\Pi$ is \emph{sound} if $\Pr_z[V^\Pi(x, z)\text{ accepts}] > 1/2$. By the {\bf Soundness} item of \autoref{thm: PCP}, if $\Pi$ is any sound PCP proof that maximizes $\widetilde{\pfx}(\Pi)$ (over all sound PCP proofs), then $\widetilde{out}(\Pi)$ is equal to $f(x)$. Hence, the task of $\Recon_f(x, i)$ is to find (a succinct description of) such a PCP proof $\Pi$ and output the $i$-th bit in $\widetilde{out}(\Pi)$.
    
    Note that whenever $\Solve_f(x, G)$ does not succeed, by \autoref{thm: Jerabek-Korten}, the circuit complexity of $\Pi_x$ is at most $k'(n) := O(k(n)^4\log t(n))$. The algorithm $\Recon_f(x, i)$ proceeds in $r'(n)$ rounds. It maintains a string $\pfx$ (initially the empty string). In each round, it asks the following $\pr\MA$ query $q_{\pfx}$:
    \begin{itemize}
        \item Merlin provides a circuit $C:\{0, 1\}^{\log|\Pi_x|}\to\{0, 1\}$ of size at most $k'(n)$.
        \item Arthur samples randomness $z\gets O(\log t(n))$ (for the PCP verifier) and $z_{\Corr} \gets \{0, 1\}^{\polylog(t(n))}$ (for the \hyperref[thm: low-degree correction]{Low-Degree Self-Corrector}).
        \item If $V^C(x, z)$ rejects then Arthur rejects. Otherwise, let $\Pi_C$ be the truth table of $C$ (which is a PCP proof), Arthur accepts if and only if $\pfx\circ 1$ is a prefix of $\widetilde{\pfx}(\Pi_C)$ (where each bit of $\widetilde{\pfx}(\Pi_C)$ is evaluated by the \hyperref[thm: low-degree correction]{Low-Degree Self-Corrector}).
    \end{itemize}

    If $q_\pfx$ accepts, then we let $\pfx \gets \pfx\circ 1$; otherwise we let $\pfx \gets \pfx\circ 0$. Then we proceed to the next round. At the end of $r'(n)$ rounds, we obtain a string $\pfx \in \{0, 1\}^{r'(n)}$. We make the final query $q_{\pfx, i}$, which is the same query as $q_{\pfx}$ except that Arthur additionally rejects if the $i$-th bit $\widetilde{out}(\Pi_C)$ is $0$ (again, we access $\widetilde{out}(\Pi_C)$ using the \hyperref[thm: low-degree correction]{Low-Degree Self-Corrector}). If $q_{\pfx, i}$ accepts, then $\Recon_f(x, i)$ outputs $1$; otherwise $\Recon_f(x, i)$ outputs $0$.

    Clearly, $\Recon_f(x, i)$ runs in $\poly(k(n), r(n), \log t(n))$ time and the size of each query $q_\pfx$ and $q_{\pfx, i}$ is also upper bounded by $\poly(k(n), r(n), \log t(n))$.

    Assuming $\Solve_f(x, G)$ does not succeed, we show that every query made by $\Recon_f(x, i)$ is smart (i.e., in the promise of $\pr\MA$). Recall that $\Pi_x$ is the PCP proof for $x$ guaranteed in the {\bf Completeness} item of \autoref{thm: PCP}. We first use an induction to show that $\pfx$ (i.e., the string finally obtained in $\Solve_f(x, G)$) is equal to the length-$r'(n)$ prefix of $\Pi_x$. For notation convenience, let $\pfx_{\le j}$ denote the length-$j$ prefix of $\pfx$, hence the $\pr\MA$ query made in the $j$-th round is $q_{\pfx_{\le j-1}}$. Let $j < r'(n)$, suppose that after $j$ rounds, $\pfx$ is equal to the length-$j$ prefix of $\Pi_x$ (this is clearly true for $j = 0$), we now consider the $(j+1)$-th round:
    \begin{itemize}
        \item If the $(j+1)$-th bit of $\Pi_x$ is $1$, then in $q_{\pfx_{\le j}}$, there exists a response of Merlin that makes Arthur accept with probability $1$. Hence, $q_{\pfx_{\le j}}$ satisfies the promise of \textsc{Yes} instances of $\pr\MA$.
        
        To see this, let $C$ be the size-$k'(n)$ circuit whose truth table is $\Pi_x$. Then $\Pr_z[V^C(x, z)\text{ accepts}] = 1$. Since $\Pi_x = (\tilde{w}_x, \Pi'_x)$ for some polynomial $\tilde{w}_x$ of total degree at most $m|H|$, the \hyperref[thm: low-degree correction]{Low-Degree Self-Corrector} always returns the correct value in $\widetilde{\pfx}(\Pi_x)$ and $\widetilde{out}(\Pi_x)$ with probability $1$. Finally, $\pfx_{\le j}\circ 1$ is a prefix of $\widetilde{\pfx}(\Pi_x)$. It follows that Arthur accepts with probability $1$ given Merlin's proof $C$.

        \item If the $(j+1)$-th bit of $\Pi_x$ is $0$, then in $q_{\pfx_{\le j}}$, no matter what circuit Merlin provides to Arthur, Arthur rejects w.p. $\ge 1/2$. Hence, $q_{\pfx_{\le j}}$ satisfies the promise of \textsc{No} instances of $\pr\MA$.

        To see this, let $C$ be any circuit and $\Pi_C = (\tilde{w}_C, \Pi'_C)$ be its truth table (which corresponds to a PCP proof). If $\Pr_z[V^C(x, z)\text{ accepts}] \le 1/2$, then Arthur rejects with probability $\ge 1/2$. Otherwise, by the {\bf Soundness} of \autoref{thm: PCP}, $\tilde{w}_C$ is $1/4$-close to a polynomial of total degree $m|H|$, and $\widetilde{\pfx}(\Pi_C)$ is no larger than the length-$r'(n)$ prefix of $\Pi_x$, which is strictly smaller than $\pfx_{\le j}\circ 1$. Since $|\F| > \omega(|H|^2)$, the \hyperref[thm: low-degree correction]{Low-Degree Self-Corrector} is always correct with high probability (the correct probability can be boosted to $1-1/\poly(t(n))$ by repetition), which means that Arthur will detect that $\pfx_{\le j}\circ 1$ is not a prefix of $\widetilde{\pfx}(\Pi_C)$ with high probability and hence reject.
    \end{itemize}
    It follows that the $(j+1)$-th $\pr\MA$ query satisfies the $\pr\MA$ promise, and it is a \textsc{Yes} instance if and only if the $(j+1)$-th bit of $\Pi_x$ is $1$. Therefore, after $r'(n)$ rounds, $\pfx$ is equal to the length-$r'(n)$ prefix of $\Pi_x$.

    Similarly, we can show that the final query $q_{\pfx, i}$ is inside the $\pr\MA$ promise and is a \textsc{Yes} instance if and only if the $i$-th bit of $f(x)$ is $1$. The proof is very similar so we only provide a sketch:\begin{itemize}
        \item If the $i$-th bit of $f(x)$ is indeed $1$, then Arthur accepts with probability $1$ when Merlin sends a small circuit whose truth table is $\Pi_x$.
        \item If the $i$-th bit of $f(x)$ is $0$, then for every PCP proof $\Pi$, either $\Pr_z[V^C(x, z)\text{ accepts}] \le 1/2$, or $\widetilde{\pfx}(\Pi)$ is strictly smaller than the length-$r'(n)$ prefix of $\Pi_x$, or $\widetilde{\pfx}(\Pi)$ is equal to that prefix. In the first two cases, Arthur rejects w.p.~$\ge 1/2$. In the last case, by the {\bf Soundness} item of \autoref{thm: PCP}, we have $\widetilde{out}(\Pi) = f(x)$, hence the $i$-th bit of $\widetilde{out}(\Pi)$ is $0$ and Arthur rejects w.h.p.~as well.
    \end{itemize}

    In conclusion, when $\Solve_f(x, G)$ does not succeed, $\Recon_f(x, i)$ outputs the $i$-th bit of $f(x)$ in deterministic $\poly(k(n), r(n), \log t(n))$ time with smart access to a $\pr\MA$ oracle. This finishes the proof.
\end{proof}

\begin{remark}\label{remark: instance-wise HvR}
    \autoref{thm: instance-wise HvR} is an \emph{instance-wise hardness-randomness tradeoff} for solving $\Avoid$ in the following sense. Given any multi-output function $f$ computable in a certain deterministic class $\calA$ that is almost-all-input hard against a certain randomized class $\calB$, we can solve the range avoidance problem in a related deterministic class $\calC$. The statement of \autoref{thm: instance-wise HvR} implies such an instance-wise hardness-randomness tradeoff where
    \begin{itemize}
        \item $\calA = \TIME[t(n)]^\NP\RoundLength{r(n)}{s(n)}$,
        \item $\calB = \smart\TIME[\poly(k(n), r(n), \log t(n))]^{\prMA}$, and
        \item $\calC = \TIME[\poly(t(n))]^\NP\RoundLength{r(n)s(n) + O(k(n) + \log t(n))}{\max\{s(n), k(n)\}}$.
    \end{itemize}

    It is instructive to compare this statement to that of~\cite{CT21b}. In~\cite{CT21b} it was proven that given any multi-output function computable in a certain deterministic class $\calA$ that is almost-all-input hard against a certain randomized class $\calB$, we can solve $\Gap\SAT$ (the canonical $\pr\RP$-complete problem) in a related deterministic class $\calC$, where:
    \begin{itemize}
        \item $\calA$ is the class of logspace-uniform circuits of depth $d$ and time $T$;
        \item $\calB$ is the class of randomized algorithms running in time $\poly(d, \log T, n)$; and
        \item $\calC$ is the class of deterministic algorithms running time $\poly(T, n)$.
    \end{itemize}
\end{remark}

\subsection{The Lower Bound}

Fix a $\P$-uniform family of $\Avoid$ instances $\{G_n\}_{n=2^k,k \in \N}$, where for each $n = 2^k$ a power of $2$, $G_n: \{0, 1\}^n \to \{0, 1\}^{2n}$ is a circuit of size at most $n^4$. We show that there exists a $\smart\FP^\MA/_1$ algorithm solving $\Avoid$ on this family infinitely often. Looking ahead, combined with the \JK~reduction from $\Avoid$ to finding hard truth tables, our algorithm will be able to solve general $\Avoid$ instances (without uniformity constraints) on infinitely many input lengths.

\begin{theorem}\label{thm: iterative win-win}
    There is an infinitely-often $\smart\FP^\prMA/_1$ algorithm $\calA$ solving $\Avoid$. In particular, there are infinitely many input lengths $\{n_i = 2^{k_i}\}$ such that $\calA(1^{n_i})$ outputs a string $y \in \{0, 1\}^{2n_i}\setminus\Range(G_{n_i})$, and for every input length $n = 2^k$, $\calA(1^n)$ never makes any query outside the $\prMA$ promise.
\end{theorem}

\begin{proof}
    Let $C \geq 3$ be a sufficiently large integer constant in the following. Let $\{n_i\}$, $\{T_i\}$, $\{r_i\}$, $\{s_i\}$ be parameters to be fixed (corresponding to input lengths, running times, number of parallel rounds, and query lengths, respectively). Let $\{G_n\}_{n = 2^k, k \in \N}$ denote a $\P$-uniform family of $\Avoid$ instances where for each $n=2^k$ a power of $2$, $G_n: \{0, 1\}^n \to \{0, 1\}^{2n}$ is a circuit of size at most $n^4$. Let $n_0 = 2^q$ be some fixed input length, let $T_0 = 2^{c n_0}$ for a universal constant $c \geq 1$, and let $r_0 = 0$, $s_0 = 0$. Let $\BF_0$ denote the {\bf Brute Force} algorithm for solving $\Avoid$ on $G_{n_0}$, which simply enumerates every string in $\{0, 1\}^{n_0}$, computes the whole range of $G_{n_0}$, and outputs the lexicographically first string not in $\Range(G_{n_0})$. Then, $\BF_0$ runs in deterministic $T_0$ time and makes no queries to the $\SAT$ oracle, so $r_0 = 0$ and $s_0 = 0$ suffice.

    Fix a universal constant $D \ge 1$ large enough so that the running time of $\Solve$ (from \autoref{thm: instance-wise HvR}) is asymptotically bounded by $t(n)^D$, and set $C$ so that $C > D + 1$. Let $n_0=2^q$ be large enough that $2^{c n_0} \ge n_0^C$. For each $i \ge 0$, set
    \[ n_{i+1} := n_i^C \qquad\text{and}\qquad  T_i := \max\{2^{cD^i n_0}, n_{i+1}\},\] and note that by our choice of $n_0$, our definition of $T_i$ is consistent with our earlier definition of $T_0 = 2^{c n_0}$.
    
    We inductively define an algorithm $\BF_i$ for the instance $G_{n_i}$. At stage $i$, fix $x_0 := 0^{n_0}$ and let $f_i: \{0, 1\}^{n_0} \to \{0, 1\}^{T_i}$ be the function that ignores its input and outputs the $2n_i$-bit string produced by $\BF_i$ on $G_{n_i}$, padded with extra zeros so that the total output has length $T_i$. Since $T_i \ge 2n_i$ and $\BF_i$ runs in $\TIME[T_i]^\NP\RoundLength{r_i}{s_i}$, the function $f_i$ can also be implemented in time $O(T_i)$ with $r_i$ rounds of $\NP$ queries with witness length $s_i$. Also, since $T_i \ge n_{i+1}$ and $|G_{n_{i+1}}| \le n_{i+1}^4$,  \autoref{thm: instance-wise HvR} applies to $f_i$ with $k(n_0) = n_{i+1}$. (Technically, we need $f_i$ to be defined on all inputs; we can simply have $f_i$ output all-zeroes on inputs that are not of length $n_0$.) From the theorem, we obtain:
    \begin{itemize}
        \item an algorithm $\BF_{i+1} := \Solve_{f_i}(x_0, G_{n_{i+1}})$ that attempts to solve $\Avoid$ on $G_{n_{i+1}}$ and runs in
        \[\TIME[T_i^D]^\NP\RoundLength{r_{i+1} := r_is_i + O(n_{i+1} + \log T_i)}{s_{i+1} := \max\{s_i, n_{i+1}\}}, \text{as well as}
        \]
        \item an algorithm $\Recon_i(j) := \Recon_{f_i}(x_0, j)$ that runs in deterministic $\poly(n_{i+1}, r_i, \log T_i)$ time with access to a $\prMA$ oracle. If $\BF_{i+1}$ fails to solve $\Avoid$ on $G_{n_{i+1}}$, then for every $j \in [T_i]$, $\Recon_i(j)$ outputs the $j$-th bit of $f_i(x_0)$ while making only smart $\prMA$ queries.
    \end{itemize}
    By our choice of $C$ and $D$, the running time of $\BF_{i+1} = \Solve_{f_i}(x_0, G_{n_{i+1}})$ is bounded from above by \[T_i^D = \max\{2^{cD^{i+1}n_0}, n_{i+1}^D\} \le \max\{2^{cD^{i+1}n_0}, n_{i+1}^C\} = \max\{2^{cD^{i+1}n_0}, n_{i+2}\} = T_{i+1}.\] Therefore $T_{i+1}$ is a time upper bound for $\BF_{i+1}$.

    Let us bound the other parameters. First, we observe that $s_i = n_i$ for every $i \ge 1$, since $s_0 = 0$ and $s_{i+1} = \max\{s_i, n_{i+1}\}$. Second, we claim that $r_i \le O(n_i)$ for every $i \ge 1$. Indeed, let $c'\ge 1$ be a constant such that $r_{i+1} = r_is_i + c'\cdot (n_{i+1} + \log T_i)$, then we can inductively prove that $r_i \le 3c'n_i$ for every $i\ge 0$. This is because $r_0 = 0$, and if $r_i \le 3c'n_i$ then
    \[
        r_{i+1} \le r_is_i + c'\cdot (n_{i+1} + \log T_i)
        \le 3c'\cdot n_i^2 + c'\cdot (n_{i+1} + \log T_i)
        \le 3c'\cdot n_{i+1},
    \]
    because $s_i = n_i$, $C \ge 3$ (which means $n_i^2 \le o(n_{i+1})$), and $\log T_i \le O(D^i n_0 + \log n_{i+1}) \le n_{i+1}$ for all sufficiently large $n_0$. Hence one call to $\Recon_i(j)$ takes $\poly(n_i)$ time with a $\prMA$ oracle, and reconstructing all first $2n_i$ bits of $f_i$ takes $\poly(n_i)$ time with a $\prMA$ oracle. 

    Let $i_\star$ be the least index such that $2^{cD^{i_\star} n_0} \le n_{i_\star+1}$. Such an $i_\star$ exists because $C > D$ and $n_{i+1} = n_0^{C^{i+1}}$. By definition of $i_\star$, for every $i \ge i_\star$ we have $T_i = \max\{2^{cD^i n_0}, n_{i+1}\} = n_{i+1}$. To see this, note that once $2^{c D^i n_0} \leq n_{i+1}$ is true for some $i$, it is true for all larger $i$ (because $D < C$, so $2^{c D^{i+1} n_0} \le n_{i+1}^C = n_{i+2}$). Therefore $T_i = n_{i+1} = n_i^C$, so $\BF_i$ can be simulated in deterministic $\poly(n_i)$ time with oracle access to $\SAT$.

    We claim that the above sequence of input lengths $\{n_i\}$ contributes at least one ``good'' input length on which $\Avoid$ is correctly solved in polynomial time. If there is some $i < i_\star$ for which $\BF_{i+1}$ fails on $G_{n_{i+1}}$, let $i$ be the least such index. Then $\BF_i$ succeeds on $G_{n_i}$, and since $\BF_{i+1}$ fails, the reconstruction guarantee above implies that for every $j \in [2n_i]$, $\Recon_i(j)$ outputs the $j$-th bit of $f_i(x_0)$, and these first $2n_i$ bits are the output of $\BF_i$ on $G_{n_i}$. Therefore, by running $\Recon_i(1), \ldots, \Recon_i(2n_i)$, we obtain a string in $\{0, 1\}^{2n_i} \setminus \Range(G_{n_i})$ in deterministic $\poly(n_i)$ time with smart $\prMA$ queries. On the other hand, if no such $i < i_\star$ exists, then $\BF_{i_\star}$ succeeds on $G_{n_{i_\star}}$, and it runs in deterministic $\poly(n_{i_\star})$ time with oracle access to $\SAT$ by the previous paragraph.

    We now run the above construction independently with different starting input lengths $n_0$: namely, we consider all $n_0 = 2^q$ such that $q$ is not divisible by $C$. Observe that \emph{every} sufficiently large $n = 2^k$ belongs to exactly one such sequence $\{n_i\}$, obtained by writing $k = qC^i$ for $q$ not divisible by $C$. (For those finitely many $n_0 = 2^q$ that are not large enough in the analysis above, our final algorithm will simply output $0^{2n}$ and make no oracle queries; note this won't affect the infinitely-often conclusion of the theorem.)

    Using the $\P$-uniformity of $\{G_n\}$ and the uniformity of the transformations in \autoref{thm: instance-wise HvR}, we obtain two uniform polynomial-time oracle algorithms $\calA_0$ and $\calA_1$ with oracle access to $\prMA$:
    \begin{itemize}
        \item On input $1^n$, $\calA_0$ locates the (unique) sequence of inputs $\{n_i\}$ containing $n$, the corresponding $n_0$, and the index $i$ such that $n_i = n$ in this sequence. If $i \geq i_\star$ (recall $i_{\star}$ is defined as the least index such that $2^{cD^{i_\star} n_0} \le n_{i_\star+1}$), then $\calA_0$ simulates $\BF_i$ using queries to $\SAT$. Otherwise, $\calA_0$ outputs $0^{2n}$.
        \item On input $1^n$, $\calA_1$ locates the sequence of inputs $\{n_i\}$  containing $n$, runs the corresponding reconstruction algorithm $\Recon_i(j)$ for each $j \in [2n]$, and outputs the resulting $2n$-bit string.
    \end{itemize}
    Both algorithms run in deterministic polynomial time. Moreover, every query made by $\calA_0$ is automatically inside the promise of $\prMA$, while $\calA_1$ is smart exactly on those lengths for which the solver $\BF_{i+1}$ fails.

    Finally, for each input length sequence $\{n_i\}$ defined by some $n_0 = 2^q$, let $n(q)$ be some good input length in that sequence. If $n(q)$ comes from the output of $\calA_0$ then we set the advice bit $\alpha_{n(q)} := 0$, otherwise we set $\alpha_{n(q)} := 1$. On all other lengths $n$, we set $\alpha_n := 0$. We define $\calA(1^n, \alpha_n) := \calA_{\alpha_n}(1^n)$. Then $\calA$ runs in polynomial time and is smart on every power-of-two input length. Moreover, for every sufficiently large $q$ there is a distinct input length $n=n(q)$ such that $\calA$ outputs a string in $\{0, 1\}^{2n} \setminus \Range(G_n)$. Observe that for all distinct $q\neq q'$ that are large enough, the corresponding good lengths $n(q)$ and $n(q')$ must be distinct. Since there are infinitely many such $q$, there are infinitely many good input lengths. This concludes the proof.
    \end{proof}

For each $n = 2^k$, define $G_n:\{0, 1\}^n \to \{0, 1\}^{2n}$ to be the truth table generator; that is, the input of $G_n$ is a circuit $C: \{0, 1\}^{k+1} \to \{0, 1\}$ described in $n$ bits, and the output is the truth table of $C$ of length $2^{k+1} = 2n$. As shown in~\cite{Korten21}, the $\Avoid$ problem for general circuits reduces to the $\Avoid$ problem on the family $\{G_n\}$ via an $\FP^\NP$ reduction.

\def\AKorten{{\calA_{\rm Korten}}}
\begin{lemma}[{\cite[Theorem 7]{Korten21}, see also \cite[Lemma 5.14]{CHR24}}]\label{lemma: Korten reduction}
    There is an $\FP^\NP$ algorithm $\AKorten$ such that:
    \begin{enumerate}
        \item $\AKorten$ takes an $s$-size circuit $C: \{0, 1\}^n \to \{0, 1\}^{n+1}$ and a truth table $f \in \{0, 1\}^{2^k}$ as inputs, where $2^k \ge s^3$ and $n \le s$.
        \item If the circuit complexity of $f$ is at least $c_1\cdot k\cdot s$ for a sufficiently large universal constant $c_1$, then $\AKorten(C, f)$ outputs a string $y \in \{0, 1\}^{n+1} \setminus \Range(C)$.
    \end{enumerate}
\end{lemma}

Therefore, we have:
\begin{theorem}
    There is an infinitely often $\smart\FP^\prMA/_1$ algorithm $\calA$ such that for infinitely many $s \in \N$, for all $s$-size circuit $C: \{0, 1\}^n \to \{0, 1\}^{n+1}$ when $n \le s$, $\calA(C)$ outputs a string $y_C \in \{0, 1\}^{n+1}\setminus\Range(C)$.
\end{theorem}
\begin{proof}
    Let $\calA'$ denote the $\smart\FP^\prMA/_1$ algorithm in \autoref{thm: iterative win-win} for the family of truth table generators $\{G_n\}$. There are infinitely many integers $k$ such that $\calA'(1^{2^k})$ outputs a truth table of length $2^{k+1}$ that cannot be computed by circuits of size $0.1\cdot 2^k / k$.

    Given a circuit $C$ of size $s$, let $k := \lceil 4\log s\rceil$ and $f \in \{0, 1\}^{2^k}$ be the output of $\calA'(1^{2^{k-1}})$. Our algorithm $\calA(C)$ outputs $\AKorten(C, f)$. By \autoref{lemma: Korten reduction}, our algorithm is correct on infinitely many $s\in\N$.
\end{proof}

Just as in~\cite[Section 4.3]{CHLR26}, the ability to solve the Range Avoidance problem in $\smart\FP^{\prMA}/_1$ allows us to prove a number of ``near-maximum'' complexity lower bounds for $\smart\E^{\prMA}/_1$, via direct reductions to $\Avoid$.

\begin{corollary}
    The following ``near-maximum'' complexity lower bounds hold for $\smart\E^{\prMA}/_1$:
    \begin{enumerate}
        \item $\smart\E^{\prMA}/_1\not\subseteq\SIZE[2^n / n]$.
        \item Let $\delta_1, \delta_2 > 0$ be constants such that $\delta_1 + 2\delta_2 < 1$. Then there is a language $L \in \smart\E^{\prMA}/_1$ that cannot be $(1/2 + 2^{-\delta_2 n})$-approximated by size-$2^{\delta_1 n}$ circuits.
        \item For any constant $k\ge 1$ and nice function $\alpha(n) \ge \omega(1)$, $\smart\E^{\prMA}/_1\not\subseteq\TIME[2^{kn}]/_{2^n - \alpha(n)}$.
    \end{enumerate}
\end{corollary}

Finally, we note that our $\Avoid$ algorithm implies a $\smart\FP^{\prMA}/_1$ algorithm for explicit constructions of Ramsey graphs, two-source extractors, rigid matrices, optimal linear codes meeting the GV bound, etc. These follow directly from known uniform reductions to $\Avoid$~\cite{Korten21,DBLP:conf/approx/GuruswamiLW22,DBLP:conf/focs/RenSW22}.

\section{Discussion}\label{sec: conclusion}

We believe the approach of this paper has the potential to obtain near-maximum circuit lower bounds for classes such as $\pr\MA_\E$. In this section, we discuss the immediate technical obstacles we are facing and provide some speculations on future directions. (We also note that our approach is radically different from the half-exponential lower bound in~\autoref{appendix: E to prMA halfexp}. In particular, non-relativizing techniques such as the PCP theorem are crucial to our proofs.)

A significant bottleneck in our proof is the \emph{highly adaptive} search-to-decision reduction for $\SAT$: the $\JK$ reduction (\autoref{thm: Jerabek-Korten}) essentially makes $O(\log T)$ rounds of non-adaptive queries to $\searchSAT$, and the computational history of $\P^\NP\RoundLength{r(n)}{s(n)}$ (\autoref{lemma: computational history of P NP a}) is also computable in $r(n)$ rounds of non-adaptive queries to $\searchSAT$. Replacing the $\searchSAT$ oracle with a (decisional) $\SAT$ oracle incurs an $O(n/\log T)$ and $O(s(n))$ multiplicative overhead on the number of rounds, respectively. In fact, it is not hard to see from our proof that an $\FP^\NP_{tt}$ algorithm for $\searchSAT$ would imply near-maximum circuit lower bounds for smaller classes; here, $\FP^\NP_{tt}$ is the class of search problems computable by a deterministic polynomial-time algorithm with \emph{non-adaptive} access to $\NP$ oracles.

\begin{theorem}\label{thm: better searchSAT implies better lower bounds}
    Assuming there is an $\FP^\NP_{tt}$ algorithm for $\searchSAT$. Recall $\E^{\prMA[2^{\eps n}]}$ is the class of languages computable in exponential time with $2^{\eps n}$ many adaptive queries to a $\prMA$ oracle. Then, for every constant $\eps > 0$,
    \[\smart\E^{\prMA[2^{\eps n}]}/_1 \not\subseteq \SIZE[2^n / (2n)]\quad\text{and}\quad(\MA_\E\cap\coMA_\E)/_{2^{\eps n}}\not\subseteq \SIZE[2^n / (2n)].\]
\end{theorem}
We treat \autoref{thm: better searchSAT implies better lower bounds} as an invitation to the open questions (\autoref{q: HvR for searchSAT} and \autoref{q: HvR for AVOID}) that we will discuss later, rather than a main result. Therefore we only provide a proof sketch:
\begin{proof}[Proof Sketch]
    Suppose that $\searchSAT \in \FP^\NP_{tt}$. Then the round complexity in our technical lemmas can be radically reduced:
    \begin{itemize}
        \item The $\JK$ reduction (\autoref{thm: Jerabek-Korten}) runs in $\TIME[\poly(|T|, |G|)]^\NP\RoundLength{O(\log T)}{n}$.
        \item The computational history of $\TIME[t(n)]^\NP\RoundLength{r(n)}{s(n)}$ (\autoref{lemma: computational history of P NP a}), hence the valid PCP in the {\bf Completeness} case of \autoref{thm: PCP}, can be computed in $\TIME[\poly(t(n))]^\NP\RoundLength{r(n)}{s(n)}$.
        
        (Here, a minor technical detail is that the computational history $\Pi$ we find in \autoref{lemma: computational history of P NP a} is no longer the lexicographically largest one, since our $\FP^\NP_{tt}$ algorithms for $\searchSAT$ does not necessarily return the lexicographically largest satisfying assignment. However, $\Pi$ is still a valid computational history with a lexicographically largest possible length-$r'(n)$ prefix. One can also verify that the proof of~\autoref{thm: instance-wise HvR} only requires $\Pi_x$ to have the largest possible length-$r'(n)$ prefix.)
        \item Hence, in our instance-wise hardness-randomness tradeoff (\autoref{thm: instance-wise HvR}), given a multi-output function computable in $\TIME[t(n)]^\NP\RoundLength{r(n)}{s(n)}$, $\Solve$ can be computed in
        \[\TIME[\poly(t(n))]^\NP\RoundLength{r(n) + O(\log t(n))}{\max\{s(n), k(n)\}}.\]
        Recall that $\Recon$ makes $O(r(n)\log t(n))$ many queries to the $\prMA$ oracle.
        \item The parameters $r_i$ (round complexity) in our $\Avoid$ algorithm (\autoref{thm: iterative win-win}) satisfies that $r_0 = 0$ and $r_{i+1} = r_i + O(\log T_i) = r_i + O(cD^i n_0 + \log n_{i+1})$. Since $n_i = n_0^{C^i}$, by setting $C$ to be a large enough constant relative to $1/\eps$ and $D$, we will always have $r_i \le n_i^{\eps/2}$ and $\log T_i \le n_i^{\eps / 2}$ (for $i\ge 1$). Hence, each output bit of our final $\Avoid$ algorithm on input length $n_i$ can be computed in $\poly(n_i)$ time with $O(r_i\log T_i) \le (n_i)^\eps$ many adaptive queries to the $\prMA$ oracle. 
    \end{itemize}

    Let $L$ be the language defined by our $\Avoid$ algorithm, where the input instances $\{G_N: \{0, 1\}^{N/2} \to \{0, 1\}^N\}_{N = 2^n, n \in \N}$ are truth table generators~\cite{KleinbergKMP21, CHLR26}; that is, on input $i \in \{0, 1\}^n$, $L(i)$ is the $i$-th output bit of our $\Avoid$ algorithm on $G_{2^n}$. Since our $\Avoid$ algorithm is correct on infinitely many input lengths, $L\not\in\SIZE[\frac{2^n}{2n}]$. On the other hand, $L(i)$ can be computed by a deterministic $\poly(2^n)$-time algorithm with $(2^n)^\eps = 2^{\eps n}$ many adaptive and smart queries to a $\prMA$ oracle, hence $L \in \smart\E^{\prMA[2^{\eps n}]}/_1$.

    Furthermore, we claim that $L \in (\MA_\E \cap \coMA_\E)/_{2^{\eps n}}$. The reason is that in the algorithm $\Recon_f(x, i)$ in~\autoref{thm: instance-wise HvR}, only the last $\prMA$ query depends on $i$. Hence, we can simply hardwire the answers of all other queries as an advice string of length $2^{\eps n}$ and run the last query. Moreover, the last query is of the following form: Merlin sends Arthur a circuit $C$ that succinctly represents a PCP proof $\Pi_C$, Arthur verifies that $\Pi_C$ is accepted by the PCP verifier, has the correct prefix, and a certain bit of $\Pi_C$ (depending on $i$) is equal to $1$; and $L(i) = 1$ if and only if this query accepts. This immediately gives an $\MA_\E/_{2^{\eps n}}$ algorithm for computing $L$. To obtain an $(\MA_\E\cap\coMA_\E)/_{2^{\eps n}}$ algorithm, we simply change ``a certain bit of $\Pi_C$ is equal to $1$'' to this bit being equal to $b$ in Arthur's acceptance criterion, where $b \in \{0, 1\}$ is an input bit.
\end{proof}

It is known that $\searchSAT$ admits an efficient \emph{randomized} algorithm with parallel access to the $\NP$ oracle by using the isolation lemma~\cite{MulmuleyVV87, ValiantV86, Ben-DavidCGL92}. Hence, the assumption made in \autoref{thm: better searchSAT implies better lower bounds} seems plausible. However, \emph{derandomizing} this algorithm appears strictly harder than proving exponential circuit lower bounds for $\pr\MA_\E$. Indeed, we only know that $\searchSAT \in \FP^\NP_{tt}$ under lower bounds against \emph{single-valued nondeterministic} (SVN) circuits~\cite{ArvindK01, KlivansM02, MiltersenV05, ShaltielU06}. It is unclear whether such a strong derandomization assumption is necessary, and we feel that it is important to understand the \emph{minimum assumption} needed to put $\searchSAT$ into $\FP^\NP_{tt}$.

\begin{question}\label{q: HvR for searchSAT}
    What is the weakest circuit lower bound assumption that implies $\searchSAT \in \FP^\NP_{tt}$? Can we design an $\FP^\NP_{tt}$ algorithm for $\searchSAT$ under the assumption that $\E^\NP_{tt}$ requires exponential size \emph{circuits} (instead of, e.g., \emph{SVN circuits})?
\end{question}

It seems equally important, if not more, to understand the minimum assumption for solving the Range Avoidance problem in $\FP^\NP_{tt}$. By~\cite{Korten21}, the assumption that $\E^\NP$ requires exponential size circuits (\emph{not} SVN circuits!)~implies an $\FP^\NP$ algorithm for $\Avoid$; moreover, this is provably the minimum assumption. Puzzlingly, to the best of our knowledge, the weakest assumptions known to imply an $\FP^\NP_{tt}$ algorithm for $\Avoid$ are still lower bounds against SVN circuits~\cite{KlivansM02, MiltersenV05}.

\begin{question}\label{q: HvR for AVOID}
    What is the weakest circuit lower bound assumption that implies $\Avoid \in \FP^\NP_{tt}$? %
\end{question}

It is conceivable that with a better hardness-randomness tradeoff than $\JK$, we will be able to obtain a better uniform reconstruction algorithm, hence a better near-maximum circuit lower bound. Let $\Pi$ be an encoded computational history, then either $\Pi$ is ``hard enough'' to solve $\Avoid$, or $\Pi$ is ``easy'' so that one can ``reconstruct'' $\Pi$ efficiently. Is there a suitable definition of ``hardness'' so that any ``hard enough'' $\Pi$ can be used to solve $\Avoid$ in $\FP^\NP_{tt}$, while $\Pi$ being ``easy'' implies an efficient $\MA$ algorithm for reconstructing $\Pi$?

\ifnum\Anonymity=0
\section*{Acknowledgments}
We thank Lijie Chen for fruitful discussions during early stages of this work. We are grateful to Jingxun Liang, Zhenjian Lu, and Igor C.~Oliveira for helpful comments on a draft version of this paper.
\fi

\printbibliography[heading=bibintoc]

\appendix
\section{Half-Exponential Circuit Lower Bounds for \texorpdfstring{$\E^\prMA$}{E to prMA}}\label{appendix: E to prMA halfexp}

In this section, we observe that one can prove a half-exponential circuit lower bound for the class $\E^\prMA$ using relativizing techniques. This basically follows from the same argument as the Karp--Lipton collapse for $\P^\prMA$~\cite{ChakaravarthyR11}. %

\begin{theorem}
    Let $f(n)$ be a nice function such that $f(f(n^c)^c) = 2^{o(n)}$ for every constant $c\ge 1$. Then
    \[\E^\prMA\not\subseteq\SIZE[f(n)].\]
    Moreover, this holds in every relativizing world.
\end{theorem}
\begin{proof}
    We need the following two facts (both are relativizing):
    \begin{itemize}
        \item $\TIME[f(n)^{O(1)}]^{\prAM} \not\subseteq \SIZE[f(n)]$. Since $\E^\prAM \not\subseteq \SIZE[2^n / n]$, this follows from a padding argument. (Note that $\S_2\E$ contains a language with circuit complexity $\ge 2^n / n$~\cite{CHR24, Li24} and $\S_2\E\subseteq\E^{\prAM}$~\cite{ChakaravarthyR08, ChakaravarthyR11}.)
        \item If $\SAT \in \SIZE[f(n)]$, then $\prAM \subseteq \prMATIME[f(n^{O(1)})^{O(1)}]$~\cite{ArvindKSS95}.
    \end{itemize}
    There are two cases. If $\SAT \not\in \SIZE[f(n)]$, then clearly $\E^\prMA \not\subseteq \SIZE[f(n)]$ as well. Otherwise, since $\SAT \in \SIZE[f(n)]$, we have $\prAM \subseteq \prMATIME[f(n^{O(1)})^{O(1)}]$. It follows that
    \[\TIME[f(n)^{O(1)}]^{\prAM}\subseteq \TIME[f(n)^{O(1)}]^{\prMATIME[f(n^{O(1)})^{O(1)}]} \subseteq \TIME[f(f(n^{O(1)})^{O(1)})^{O(1)}]^\prMA \subseteq \E^\prMA,\]
    hence $\E^\prMA \not\subseteq \SIZE[f(n)]$.
\end{proof}
It is unclear whether such a relativizing lower bound can be proved for $\smart\E^\prMA/_1$.

\newpage
\listoffixmes

\end{document}